\documentclass[runningheads]{llncs}
\usepackage[utf8]{inputenc}
\usepackage[T1]{fontenc}
\usepackage{xspace}
\usepackage{graphicx}
\usepackage{tikz}
\usepackage{floatrow}
\usepackage{enumerate}
\usepackage{hyperref}
 \hypersetup{
  colorlinks = true
 }

\usepackage[a4paper, margin=2.5cm]{geometry}
\usepackage{amsmath,amssymb}
\usepackage{thm-restate}
\usepackage{authblk}
\usepackage{euscript}
\usepackage{multirow}
\spnewtheorem{claim}{Claim}{\bfseries}{\rmfamily}

\DeclareRobustCommand*{\ora}{\overrightarrow}
\DeclareMathOperator{\type}{\textsf{type}}
\DeclareMathOperator{\types}{\textsf{types}}
\DeclareMathOperator{\numtypes}{\#\textsf{types}}
\newcommand{\alltypes}{\mathbb{T}}
\newcommand{\ktypes}{\mathcal{K}}
\newcommand{\oH}[1]{{\ora{H_{#1}}}}
\newcommand{\oHp}[1]{{\ora{H'_{#1}}}}
\newcommand{\oHpp}[1]{{\ora{H''_{#1}}}}

\newcommand{\RR}{\mathbb{R}}

\newcommand{\G}{\mathcal{G}}
\newcommand{\D}{\mathcal{D}} 
\newcommand{\R}{\mathcal{R}} 
\newcommand{\C}{\mathcal{C}} 

\newcommand{\pD}{\partial\mathcal{D}}
\newcommand{\mS}{\mathcal{S}} 
\newcommand{\GS}{{G_\mathcal{S}}}
\newcommand{\PS}{{\ora{P_\mathcal{S}}}}
\newcommand{\PSprime}{{\ora{P_\mathcal{S'}}}}
\newcommand{\dPS}{{\ora{P^*_\mathcal{S}}}}
\newcommand{\dPSprime}{{\ora{P^*_\mathcal{S'}}}}

\newcommand{\RM}{{\mathcal{R}_M}}
\newcommand{\RMprime}{{\mathcal{R}_{M'}}}
\newcommand{\mR}{{\mathcal{R}}}
\newcommand{\mA}{{\mathcal{A}}}

\DeclareMathOperator\tw{\textsf{tw}}

\newcommand{\trh}{\textsc{Triangle Hitting}\xspace}

\newcommand{\fvs}{\textsc{Feedback Vertex Set}\xspace}

\newcommand{\FVS}{FVS\xspace}
\renewcommand{\TH}{TH\xspace}

\newcommand{\ruleref}[1]{(\hyperref[#1]{R\ref*{#1}})}

\renewcommand{\O}{\mathcal{O}}
\renewcommand{\o}{o}

\newcommand{\DEUXDIR}{2-DIR\xspace}
\newcommand{\CONTACTSEG}{contact-segment\xspace}
\newcommand{\CONTACTDEUXDIR}{contact-\DEUXDIR{}}

\def\cqedsymbol{\ifmmode$\lrcorner$\else{\unskip\nobreak\hfil
\penalty50\hskip1em\null\nobreak\hfil$\lrcorner$
\parfillskip=0pt\finalhyphendemerits=0\endgraf}\fi} 
\newcommand{\cqed}{\renewcommand{\qed}{\cqedsymbol}}

\usepackage{marvosym}
\newcommand{\apx}[1]{\hyperref[#1]{\LeftScissors}}

\graphicspath{{fig/}}

\title{Feedback Vertex Set for pseudo-disk graphs in subexponential FPT time}
\titlerunning{FVS for pseudo-disk graphs in subexponential FPT time}

\author{Gaétan Berthe\inst{1}\orcidID{0000-0003-0017-6922} \and
Marin Bougeret\inst{1}\orcidID{0000-0002-9910-4656} \and
Daniel Gonçalves\inst{1}\orcidID{0000-0003-3228-9622} \and
Jean-Florent Raymond\inst{2}\orcidID{0000-0003-4646-7602}
}

\authorrunning{G. Berthe et al.}
\institute{LIRMM, Université de Montpellier, CNRS, Montpellier, France. \and CNRS, ENS de Lyon, Université Claude Bernard Lyon 1, Inria, LIP, UMR 5668, 69342, Lyon cedex 07, France.}

\begin{document}

\maketitle

\begin{abstract}
In this paper, we investigate the existence of parameterized algorithms running in subexponential time for two fundamental cycle-hitting problems: Feedback Vertex Set (\FVS) and Triangle Hitting (\TH). We focus on the class of pseudo-disk graphs, which forms a common generalization of several graph classes where such results exist, like disk graphs and square graphs. In these graphs, we show that \TH can be solved in time $2^{\O(k^{3/4}\log k)}n^{\O(1)}$, and given a geometric representation \FVS can be solved in time $2^{\O(k^{6/7}\log k)}n^{\O(1)}$.
\end{abstract} 

\section{Introduction}
\label{sec:intro}

\paragraph{Context.}
The purpose of Parameterized Complexity is to provide an accurate view of the algorithmic complexity of a (typically NP-hard) decision problem and to understand the different contributions to the running time of the parameters of the instance.
In this framework, a standard objective is to find algorithms whose time complexities have the form $f(k)\cdot n^{\O(1)}$ where $n$ and $k$ respectively denote the size and some parameter of the instance, and $f$ is some computable function. Hence, the possibly super-polynomial part of the running time is confined to the $f(k)$ term. Such algorithms are called \emph{Fixed Parameter Tractable}, or \emph{FPT} for short.

In this paper we mainly focus on \fvs (\FVS for short) which is the problem of deciding, given a graph $G$ and an integer $k$, whether $G$ has a set of $k$ vertices whose deletion yields a forest. Some NP-hard problems like \FVS cannot be solved in time $2^{o(k)}n^{\O(1)}$ in general graphs \cite{Cygan2015Book} (assuming the Exponential Time Hypothesis) but nevertheless admit algorithms with such running times (called \emph{subexponential FPT algorithms}) when the inputs are restricted to certain graph classes.
This was initially proved for particular problems in planar graphs and related classes (like map graphs) and later unified by Demaine el al.\ \cite{bidim} in a general framework called \emph{Bidimensionality Theory}, which in essence states that every \emph{bidimensional}\footnote{Informally: yes-instances are minor-closed and a solution on the $(r,r)$-grid has size $\Omega(r^2)$.} problem (like \FVS) can be solved in subexponential FPT time on any graph class excluding a minor. 

The next step has then been to move away from minor-closed graph classes and investigate in which other classes the basic NP-hard graph problems admit subexponential FPT algorithms. Natural candidates in this direction are intersection graphs of objects in the plane. Indeed, while such graphs are not planar, the underlying planarity may allow to lift techniques and ideas from the bidimensionality theory. 
This is not straightforward in general and required new ideas as explained for example in~\cite{lokSODA22} which discusses the extension of bidimensionality to (unit) disk graphs.

These developments led to subexponential FPT algorithms in disk graphs~\cite{lokSODA22} running in time $2^{\O(k^{13/14}\log k)}n^{\O(1)}$ for \FVS, and in $2^{\O(k^{9/10}\log k)}n^{\O(1)}$ for \TH (the \trh problem where given a graph $G$ and integer $k$ one has to decide if there is a set of $k$ vertices whose deletion yields to a triangle-free graph).
Recently~\cite{Faster2023Shinwoo}, these running times have been improved  to $2^{\O(k^{7/8}\log k)}n^{\O(1)}$ for \FVS and $2^{\O(k^{2/3}\log k)}n^{\O(1)}$ for \TH when a disk representation is given, and a $2^{\O(k^{9/10}\log k)}n^{\O(1)}$ for \FVS and $2^{\O(k^{3/4}\log k)}n^{\O(1)}$ for \TH otherwise.

On the other hand, we proved in a companion paper \cite{berthe2023subexponential} that under the Exponential Time Hypothesis, neither \TH nor \FVS do admit algorithms running in time $2^{o(\sqrt{n})}$ (and thus nor in time $2^{o(\sqrt{k})}n^{\O(1)}$) in $K_{2,2}$-free contact-\DEUXDIR{} graphs of maximum degree~6, which is a very restricted subclass of pseudo-disk graphs (as shown in \autoref{sec:contactsegp-arepseudo}). In the non-parameterized setting, the best known algorithm for \FVS runs in time $2^{\Tilde{\O}(n^{2/3})}$, and actually this algorithm works for the wider class of string graphs~\cite{bonnet2019optimality}.

\paragraph{Contribution and techniques.}

In this paper we consider pseudo-disk graphs, a classical generalization of disk graphs where each vertex is now a pseudo-disk (a subset of the plane that is homeomorphic to a disk), and such that for any two intersecting pseudo-disks, their boundaries intersect on at most two points. 
We show in \autoref{sec:robust} that the algorithms of~\cite{Faster2023Shinwoo} and the analysis of their complexities can be generalized to pseudo-disk graphs with minor modifications and using some results of this paper. Here, in the case of solving the \FVS problem with a representation given as input, we provide a new approach leading to an improved time complexity of $2^{\O(k^{6/7}\log k)}n^{\O(1)}$ (\autoref{thm:FVS-pseudo}), improving on the state of the art for disk graphs.

For this, we apply the following strategy. Given an input $(G,k)$, the objective is to reduce the treewidth of $G$ to $o(k)$ in order to solve \FVS in $2^{\O(\tw(G))}n^{\O(1)}$ using a standard dynamic programming algorithm~\cite{Cygan2015Book}. To do so, we first branch using techniques of~\cite{lokSODA22} to reduce our input, to a collection of pseudo-disk graphs with bounded clique number (or equivalently with \emph{ply} at most $p$). Then we reduce the sizes of these instances to kernels of size $\O(p^4k)$. This implies, using a bound on the number of edges in a pseudo-disk graph of bounded clique number, that these kernels have  treewidth $\O(\sqrt{pn})=\O(\sqrt{p^5k})$.

\paragraph{Organization of the paper.}
In \autoref{sec:prelim} we give the necessary definitions and properties and present the aforementioned preprocessing step. \autoref{sec:fvspseu} reduces our main result to a technical kernelization lemma, that is proved in \autoref{sssec:lemma}.
\autoref{sec:contactsegp-arepseudo} is devoted to a subclass of pseudo-disk graphs, contact graphs of segments, for which the analysis of the algorithm provides a better running time.
In \autoref{sec:robust}, it is shown that the approach of~\cite{Faster2023Shinwoo} provides subexponential time algorithms for \FVS and \TH, also when the input graph is a pseudo-disk graph. We conclude with directions for future research in \autoref{sec:disc}.

\section{Preliminaries}
\label{sec:prelim}

\subsection{Basics}
In this paper logarithms are binary and all graphs are simple, loopless and undirected.
Unless otherwise specified we use standard graph theory terminology, as in \cite{diestel2005graph} for instance.
Given a graph $G$, we denote by $\omega(G)$ the maximum order of a clique in $G$.
We denote by $d_G(v)$ the degree of $v \in V(G)$, or simply $d(v)$ when $G$ is clear from the context. We use the notation $\Delta(G)$ for the maximum degree of the vertices of~$G$.
We denote by $\tw(G)$ the treewidth of $G$.

\subsection{Graph classes}
\label{sec:graphCl}

 \begin{figure}[!ht]
\centering           
    \includegraphics[width=\textwidth]{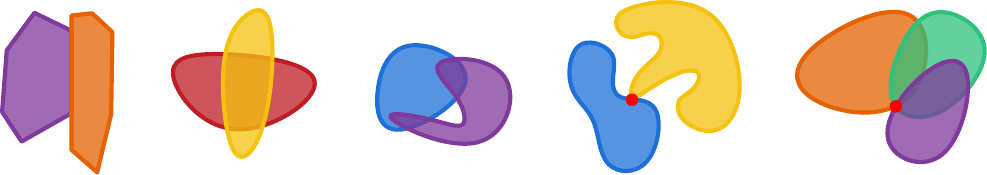}
    \caption{Three forbidden intersections between pseudo-disks, and two avoidable examples.}
    \label{fig:examples}
\end{figure}

 In this article,  we are mainly concerned with geometric graphs described by the intersection or contact of objects in the Euclidean plane.
 The most general class we consider are \emph{string graphs}, which are intersection graphs of strings (a.k.a. Jordan arcs).

\paragraph{Pseudo-disk graphs.} The focus on this paper is on \emph{pseudo-disk graphs}.
A \emph{pseudo-disk} $\D$ is a subset of the plane that is homeomorphic to a disk. We denote its boundary by $\pD$ and call \emph{internal points} any point of $\D$ that does not belong to $\pD$. A set of pseudo-disks $\mS$ forms a \emph{system of pseudo-disks} if, for any two intersecting elements $\D_1$, $\D_2\in \mS$,
their borders, $\pD_1$ and $\pD_2$, intersect on at most two points.
Under minor perturbation, any system of pseudo-disk, can be such that
for any two intersecting elements $\D_1, \D_2\in \mS$, either their borders do not intersect but in that case one is contained in the other, or their borders intersect on exactly two points while $\D_1 \cap \D_2$ contains internal points.
Similarly, we can require that no point belongs to more than two boundaries (see \autoref{fig:examples}). Given a system of pseudo-disks $\mS$, we denote by $\GS$ the corresponding intersection graph, and this defines the class of pseudo-disk graphs.

Note that pseudo-disk graphs are in particular string graphs and they form a common generalization of various classes of ``fat'' intersection graphs such as disk graphs, intersection graphs of axis-parallel squares, and more generally any intersection graph obtained from homothetic copies of a given convex set, but they also generalize the contact graphs of segments (see \autoref{sec:contactsegp-arepseudo}).

Given $\mS$ a system of pseudo-disks and $z$ a point in the plane, the \emph{ply} of $z$ (with respect to $\mS$) is the number of pseudo-disks of $\mS$ containing $z$. The \emph{ply} of a maximal connected region $\R$ of $\mathbb{R}^2\setminus \bigcup_{\D\in \mS} \partial \D$, is the ply of its points. The \emph{ply} of $\mS$ is the maximum ply of a point of the plane with respect to $\mS$. A pseudo-disk graph $G$ has \emph{ply} $p$ if it is the intersection graph of a system of pseudo-disks of ply $p$.

\paragraph{Representation of pseudo-disk graphs.} A system of pseudo-disks $\mS$ is represented by the directed plane multi-graph $\PS$ defined as follows (see \autoref{fig:primal-dual-rep}). For any two $\D_1, \D_2\in \mS$, every point in $\pD_1\cap \pD_2$ is a vertex of $\PS$. For any $\D\in \mS$, the Jordan arcs in $\pD$ joining any two such points form the arcs of $\PS$. Those arcs are oriented in such a way that $\pD$ corresponds to a clockwise cycle around $\D$.
It may remain disks $\D\in \mS$ with uncrossed boundaries (i.e. $\pD\cap \pD'=\emptyset$ for any $\pD'\in \mS$). For such a disk, we pick an arbitrary point in $\pD$ as a vertex for $\PS$, and the rest of $\pD$ corresponds to a clockwise loop on this vertex. Note that $\PS$ has at most $2|E(\GS)| + |V(\GS)|$ vertices, and at most $4|E(\GS)|+|V(\GS)|$ edges. This is a convenient feature of pseudo-disk systems to admit a polynomial (in terms of $\GS$) space data structure representing them.
This graph $\PS$ is called the \emph{pseudo-disk representation} of $\mS$ and $\GS$. We denote $\dPS$ the dual graph of~$\PS$. Observe that any arc is oriented from a region (of $\mS$) with lower ply towards one with higher ply.

\begin{figure}
    \centering
    
    \includegraphics[width=\textwidth]{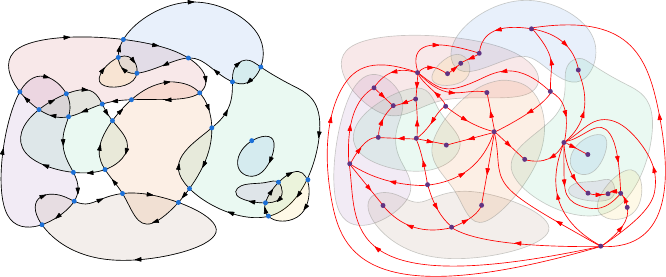}
    \caption{A pseudo-disk representation $\PS$ and its dual $\dPS$.}
    \label{fig:primal-dual-rep}
\end{figure}

for several useful operations, updating a pseudo-disk representation is simple.
\begin{lemma}\label{lem:tech-pseudo-disk-del-contr}
    Given a pseudo-disk graph $G$, deleting any vertex, or any edge $uv$ that does not belong to a triangle, or contracting any edge $uv$ that does not belong to a triangle results in a pseudo-disk graph. 
    Furthermore, given a representation of $G$ of ply $p$, one can update the representation in linear time and the new representation has ply at most $p$. 
\end{lemma}
\begin{proof}
The deletion of a vertex $v$ is obtained by simply removing $\D_v$ from the representation.

The deletion of an edge $uv$ that does not belong to any triangle is obtained by replacing 
$\D_u$ and $\D_v$, by $\D_u \setminus \D_v$ and by $\D_u \setminus \D_v$,\footnote{To be precise, $\D_u\setminus \D_v$ is not closed, and hence not homeomorphic to a closed disk.
Here and later, let us use this notation to denote a closed subset homeomorphic to a disk that is arbitrarily close to $\D_u\setminus \D_v$ while remaining inside this set.}
respectively. As we start from a system of pseudo-disks, these subsets are homeomorphic to disks. Furthermore, as $\D_u\cap \D_v$ does not intersect any other pseudo-disk $\D_w$, for $w\in V(G)\setminus \{ u, v \}$,
we have that $\pD_w$ intersect the new $\pD_u$ (resp. $\pD_v$) as many times as the former one, that is zero or twice. So, this system is a system of pseudo-disks.

The contraction of $uv$ is obtained by replacing 
$\D_u$ and $\D_v$ by $\D_u \cup \D_v$, and let us denote $w$ the new vertex.
As we start from a system of pseudo-disks $\D_w = \D_u\cup \D_v$ is homeomorphic to a disk. Now, consider any vertex $t\in N(u)\setminus \{v\}$ and let us prove that $\pD_t$ and $\pD_w$ intersect at 0 or 2 points. The case of a vertex $t'\in N(v)\setminus \{u\}$ is handled symmetrically. We divide the closed curves $\pD_u$ and $\pD_v$ into two arcs each, as follows: $\mA_u^{int}= \pD_u \cap \D_v$,
$\mA_u^{ext}= \pD_u \setminus \D_v$,
$\mA_v^{int}= \pD_v \cap \D_u$, and
$\mA_v^{ext}= \pD_v \setminus \D_u$. Note that as $t\notin N(v)$, we have that none of $\mA_u^{int}$, $\mA_v^{int}$ or $\mA_v^{ext}$ intersects $\pD_t$. Thus, $\pD_t$ intersects $\mA_u^{ext}$ in either 0 or 2 points (the first case occurs if $\D_w \subsetneq \D_u$). This implies that $\pD_t$ intersects $\pD_w = \mA_u^{ext} \cup \mA_v^{ext}$ in either 0 or 2 points.
As the other intersections of the system, not involving $\D_w$, did not change, the obtained system is a system of pseudo-disks.

Given a pseudo-disk representation $\PS$ of $G$ with ply $p$, in any of the above three operations, the update is very simple. Actually, for the edge deletion and the edge contraction, these could even be performed in constant time. It is also clear that, in any case, the ply can only decrease.
\qed
\end{proof}

\subsection{Properties of pseudo-disk graphs}

Let us present two results, of independent interest, that are needed for our approach. 
In string graphs, the treewidth and the number of edges can be related using the following theorem, obtained by combining the results on balanced separators of string graphs of Lee \cite{lee2016separators} and the links between separators and treewidth of Dvo{\v{r}}{\'a}k and Norin \cite{DVORAK2019137}.

\begin{theorem}[\cite{lee2016separators} and  \cite{DVORAK2019137}] \label{th:sep-m-tw}
    Any $m$-edge string graph has treewidth $\O(\sqrt{m})$.
\end{theorem}
 
To bound the number of edges, we use the following result:
\begin{restatable}[See the proof of Lemma 16 in \cite{lokshtanov2024linearedges}]{lemma}{apxlempseudonumberedges}\label{lem:pseudo-number-edges}
    There exists a constant $c$ such that any graph $G$ on $n$ vertices admitting a pseudo-disk representation with ply $p$ has at most $cpn$ edges.
\end{restatable}
Together with \ref{th:sep-m-tw} this gives us:
\begin{corollary}\label{cor:pseudo-tw}
Any pseudo-disk graph on $n$ vertices admitting a pseudo-disk representation with ply $p$ has treewidth $\O(\sqrt{pn})$.
\end{corollary}

\autoref{lem:pseudo-ply-clique} ensures us that it is easy to provide a constant approximation of the clique number for a pseudo-disk graph given a representation.
\begin{restatable}{lemma}{apxlempseudoplyclique}\label{lem:pseudo-ply-clique}
    There is a constant $c$ such that for every graph $G$, if $G$ admits a pseudo-disk representation with ply $p$, then $p\le \omega(G)\le cp$. This implies that given a pseudo-disk representation, there is a polynomial time algorithm that returns a clique whose size is within a constant factor from the optimum.
\end{restatable}

\begin{proof}
    The lower-bound is trivial so we discuss the upper-bound. It is clear that a maximum clique in $G$ admits a ply $p$ representation.
    Thus, by \autoref{lem:pseudo-number-edges} applied to this subgraph, ${\omega(G)(\omega(G)-1)}/{2} \le cp\omega(G)$ for some constant $c$, which leads to $\omega(G) \le (2c +1)p$.
    Hence, the clique formed by all the disks that contain some fixed region of maximum ply has at least $\omega(G)/c'$ vertices with $c'=2c+1$.
    
    In order to compute such a clique given a representation $\PS$ we proceed as follows. Recall that $\dPS$ is the dual of $\PS$ (see \autoref{fig:primal-dual-rep}). First, we construct a spanning tree $T$ of $\dPS$ (ignoring the orientations of the edges). To compute the ply of each region, it suffices to notice that for any arc $rr'$ of $T$, the ply of $r'$ is one more than the one of $r$, and that the minimum ply among the regions of $\mathcal{S}$ (i.e. the vertices of $T$) equals zero. Then, for a region $r$ with maximum ply, to output the list of pseudo-disks containing it, one can consider a path $P$ in $T$ from $r$ to a region of ply zero, and select the pseudo-disks $\D_v$ such that $\pD_v$ is crossed an odd number of times by $P$. \qed
\end{proof}

\subsection{Preliminary branching step}

As first step of our algorithms we make use of two preprocessing routines: the first one is a folklore branching that allows to deal with cliques larger than a chosen size $p$ (where typically $p=k^{\epsilon}$) and the second, allows considering a (non necessarily optimal) feedback vertex set $M$, that intersects every triangle at least twice. These steps are described in \cite{lokSODA22} for disk graphs. They are summarized in the following routine.

\begin{corollary}[Corollary 10 in~\cite{berthe2023subexponential}]\label{cor:bothbranchings}
Let $\G$ be a hereditary graph class where the maximum clique can be $\alpha$-approximated for some constant factor $\alpha\geq 1$ in polynomial time. 
There exists a $2^{\O\left (\frac{k}{p}\log k \right )}n^{\O(1)}$-time algorithm that, given an instance $(G,k)$ of $\FVS$ and an integer $p\in [6\alpha,k]$, where $G \in \G$, returns a collection
$\C$ of size $2^{\O\left (\frac{k}{p}\log k \right)}$ of tuples $(G',M,k')$ such that:
\begin{enumerate}
    \item For any $(G',M,k') \in \C$, $(G',k')$ is an instance of \FVS where $G'$ is an induced subgraph of $G$, $\omega(G') \le p$, and $k' \le k$;
    \item $|M| = \O(p k)$, and every triangle of $G'$ has at least two vertices in $M$; and
    \item $(G,k)$ is a yes-instance of \FVS if and only if there exists $(G',M,k') \in \C$ such that $(G',k')$ is a yes-instance of \FVS.
\end{enumerate}
\end{corollary}

\autoref{lem:pseudo-ply-clique} ensures us that \autoref{cor:bothbranchings} applies to pseudo-disk graphs.
As in this procedure, every graph $G'$ in the collection $\C$ is an induced subgraph of $G$, \autoref{lem:tech-pseudo-disk-del-contr} ensures us that we can also get a pseudo-disk representation $\overrightarrow{P'}$ for each of them. Note that since every graph $G'$ in $\C$ has clique number at most $p$, by \autoref{lem:pseudo-ply-clique}, its representation $\overrightarrow{P'}$ has ply at most $p$.

\section{Hitting cycles in pseudo-disk graphs}
\label{sec:fvspseu}
 
The main result of this paper is the following.

\begin{theorem}\label{thm:FVS-pseudo}
    There is an algorithm that, given a $n$-vertex pseudo-disk graph with a representation and a parameter $k$, solves \FVS in time $2^{\O(k^{6/7} \log k)}n^{O(1)}$.
\end{theorem}
To achieve this, we proceed in three steps. We first use the preprocessing of \autoref{cor:bothbranchings}. Then, we kernelize each instance provided by the branching process in \autoref{lemma:red-fvs-pseudo}. We show that these kernelized instances have small treewidth, and thus we can conclude with a dynamic algorithm to solve each of them. 
The main technical ingredient of this proof is the following lemma, whose full proof is deferred to \autoref{sssec:lemma}.

\begin{lemma}\label{lemma:red-fvs-pseudo}
Given a quadruple $(G, \PS, M,k)$ as given by \autoref{cor:bothbranchings}, there is a polynomial time algorithm that returns an equivalent instance 
$(G',k')$ where $k'\leq k$ and $G'$ is a pseudo-disk graph with ply at most $p$ and $\O(p^4 k)$ vertices.
\end{lemma}
As we will apply \autoref{cor:bothbranchings} for $p=\O(k^{\epsilon})$, the above will directly imply that 
$|V(G')|$ is almost linear, and, using \autoref{cor:pseudo-tw}, that $G'$ has sublinear treewidth. Thus, the whole technical difficulty is to prove \autoref{lemma:red-fvs-pseudo}. 
Now that the lemma is stated, we can proceed with the proof of the theorem.

\begin{proof}[of \autoref{thm:FVS-pseudo}]
We apply \autoref{cor:bothbranchings} on an instance $(G,k)$ provided with a representation $\PS$ and on a value $p\in [6,k]$ to be set later. As a result we obtain a collection $\mathcal{C}$ of $2^{\O\left (\frac{k}{p}\log k \right )}$ instances that are each provided with a ply $p$ pseudo-disk representation, and with a feedback vertex set $M$ of size $\O(pk)$ that intersects every triangle on at least two vertices. Furthermore, solving these instances of \FVS is enough to get a solution to our initial instance $(G,k)$. This first stage is done in time $2^{\O\left (\frac{k}{p} \log k\right )}n^{O(1)}$.

Then for each of the $2^{\O\left (\frac{k}{p}\log k \right )}$ obtained instances, we apply the kernelization described in \autoref{lemma:red-fvs-pseudo}, and we get a pseudo-disk graph $G'$ with ply at most $p$ and with only $\O(p^4k)$ vertices. By \autoref{cor:pseudo-tw}, this kernel has treewidth at most $\O(\sqrt{p^5k})$, and can thus be solved in time $2^{\O(\tw(G') )}n^{\O(1)} = 2^{\O\left (p^{5/2} k^{1/2} \right )}n^{\O(1)}$ with classical algorithms~\cite{Cygan2015Book}.
The overall time complexity of this algorithm is 
\[
2^{\O{\left (\frac{k}{p} \log k \right )}}n^{\O(1)} \ +\ \ 2^{\O\left (\frac{k}{p} \log k \right )} \cdot n^{\O(1)}\cdot 2^{\O\left (p^{5/2} k^{1/2} \right ) }n^{\O(1)}.
\]
Thus setting $p=k^\frac{1}{7}$ the time complexity is indeed $2^{\O\left (k^{6/7} \log k \right )}n^{\O(1)}$. \qed
\end{proof}

\section{Proof of \autoref{lemma:red-fvs-pseudo}}\label{sssec:lemma}

Recall that in \autoref{lemma:red-fvs-pseudo} we are given a quadruple $(G, \PS, M,k)$ as computed by \autoref{cor:bothbranchings}.
The equivalent instance computed in \autoref{lemma:red-fvs-pseudo} will be obtained using reduction rules.
Before applying these rules, we define the set $M'$ by adding to $M$ every pseudo-disk $\D$ that contains a point of $\pD_u\cap \pD_v$, for $u,v\in M$ (see \autoref{fig:M'+Hu}, left).
Given a vertex set $X$, we denote $\mathcal{R}_X$ the subset of the plane defined by $\left( \cup_{u \in X} \D_u\right)$.
Let us now partition the vertices of $V(G)\setminus M'$ into three types.
\begin{figure}[ht]
    \centering
    \includegraphics[scale=0.5]{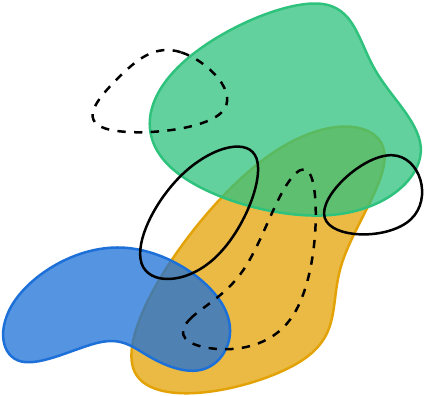}
    \hspace{1cm}
    \includegraphics[scale=0.4]{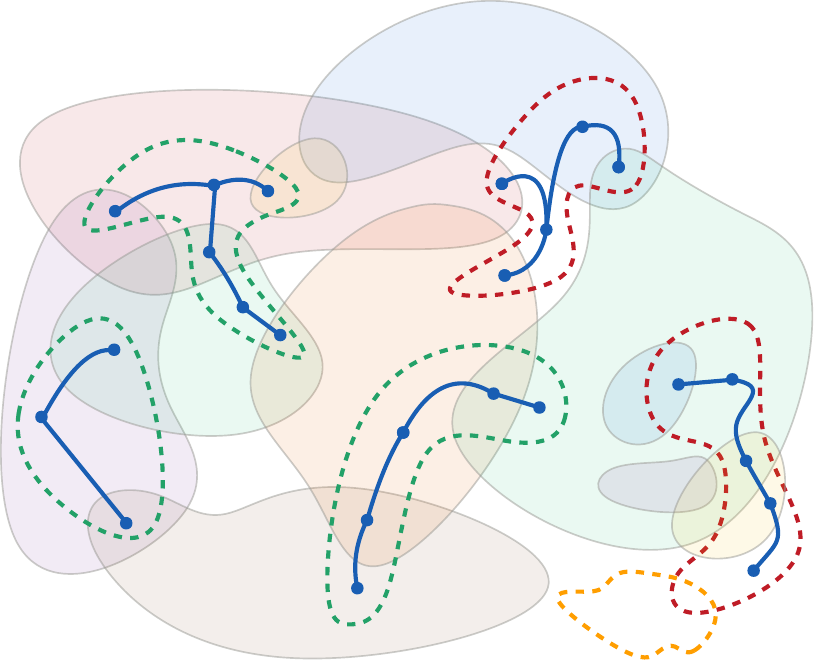} 
    \caption{Left: From $M$ to $M'$. The pseudo-disks of $M$ are filled. Among the others, those added to $M'$ have solid border. Right: Here, the filled pseudo-disks belong to $M'$. In dashed green, the inner pseudo-disks of $I_{M'}$. In dashed red, the border pseudo-disks of $B_{M'}$. In dashed orange, the only pseudo-disk $v \in O_{M'}$. For any pseudo-disk $u$ not in $M'$, we depicted its hosted graph $H_u$, which is defined in \autoref{def:Ht-intro}. 
    }
    \label{fig:M'+Hu}
\end{figure}

\begin{definition}\label{def:in-bo-out-M'}
    A vertex $v \in V(G)\setminus M'$ is:
    \begin{itemize}
        \item an \emph{inner-$M'$-vertex} if $\D_v \subseteq \RMprime$,
        \item a \emph{border-$M'$-vertex} if $\D_v$ intersects $\RMprime$ without being included, and 
        \item an \emph{outer-$M'$-vertex} if $\D_v$ does not intersect $\RMprime$.
    \end{itemize}  
    We denote $I_{M'}$, $B_{M'}$, and $O_{M'}$ the vertex sets with these three types of vertices (see \autoref{fig:M'+Hu}).
\end{definition}

Let us state a few properties about these vertices.
\begin{claim}\label{cl:pseudo-disk-abc}\mbox{}
\begin{enumerate}[(a)]
    \item \label{e:cc} $|M'|= \O(p|M|)=\O(p^2k)$.
    \item \label{e:aa} For any edge $uv$ of $G-M'$, the intersection $\D_u \cap \D_v$ does not intersect $\RMprime$.    
    \item \label{e:bb} For every $u\in V(G) \setminus M'$, the pseudo-disk $\D_u$ does not contain any point of $\pD_v\cap \pD_w$, for any two vertices  $v,w\in M'$.
    \item \label{e:dd} In $G-M'$ the vertices of $I_{M'}$ are isolated (i.e. $\forall v\in I_{M'}$ we have that $N(v) \subseteq M'$).
    \item \label{e:ee} For every non-empty tree $T$ of $G - M' - I_{M'}$, the set $\mathcal{R}_T \setminus \RMprime$ is non-empty and connected, with $\mathcal{R}_T = \left(\cup_{v\in V(T)}\D_v\right)$.

\end{enumerate}
\end{claim}
\begin{proof}
Property~\eqref{e:cc} is a consequence of the following facts. 
By assumption of the statement we aim to show (\autoref{lemma:red-fvs-pseudo}), $G$ is a pseudo-disk graph with ply at most $p$, thus \autoref{lem:pseudo-number-edges} implies that $G[M]$ has $\O(p|M|)$ edges. Hence, in the representation of $G$ there are at most $\O(p|M|)$ 
intersection point of the borders of two pseudo-disks $\D_1,\D_2\in M$.
Then, since each pseudo-disk in $M'\setminus M$ contains such an intersection point (by definition of $M'$), and since each of those points belongs to at most one pseudo-disk $\D\notin M$ (otherwise there would be a triangle with at most one vertex in $M$), we have that $|M'\setminus M| = \O(p|M|)$. Hence, $|M'| = \O(p|M|)$.

We prove property~\eqref{e:aa} by contradiction: assume that for some edge $uv$ of $G-M'$, the intersection $\D_u \cap \D_v$ intersects $\RMprime$. In that case there is a triangle $uvw$ with at most one vertex, $w$, in $M\subseteq M'$, a contradiction.

Property~\eqref{e:bb} follows from the definition of $M'$.
Property~\eqref{e:dd} follows directly from the definition of inner-$M'$-vertices and Property~\eqref{e:aa}.

For Property~\eqref{e:ee}, suppose towards a contradiction that some connected component $\mR^*$ of $\mathcal{R}_T \cap \RMprime$ is bordered by a sequence of Jordan arcs $\partial\mR^*_1,\ldots,\partial\mR^*_t$ with $t\ge 4$, and such that $\partial\mR^*_i$ is either a sub-arc of $\partial\mR_T$, if $i$ is odd, or a sub-arc of $\partial\mR_{M'}$, if $i$ is even.
Note that by Property~\eqref{e:aa}, for any even $i$ the arc $\partial\mR^*_i$ is contained inside $\D_u$ for some $u\in V(T)$, but it does not intersect any other pseudo-disk of $T$.
For any two even values $i,j$, let $\mA \subset \mR^*$ be
a Jordan arc linking $\partial\mR^*_i$ and $\partial\mR^*_j$. 
Among the possible  Jordan arcs, let $\mA$ be one minimizing the number of crossings with borders of $M'$. 
Let us denote $p_i$ and $p_j$ the endpoints of $\mA$, and let us prove that they belong to the same pseudo-disk of $T$, say $\D_u$. Indeed, if $p_i$ and $p_j$ would belong to distinct pseudo-disks of $T$, say $\D_{u_1}$ and $\D_{u_2}$, then when leaving $\D_{u_1}$, the arc $\mA$ would still be in $\RMprime$ and it would also be in a distinct pseudo-disk $\D_u$ of $T$, contradicting \eqref{e:aa}. Hence, all the arcs $\partial\mR^*_i$, with $i$ even, are contained in the same pseudo-disk $\D_u$. 

Note that by Property~\eqref{e:bb}, the arc $\partial\mR^*_2$ is contained in the border $\pD_{v_2}$ of some vertex $v_2\in M'$. As $\D_u\cap \D_{v_2} \subseteq \D_u\cap\mR_T$, we have that $\pD_u$ and $\pD_{v_2}$ intersect in at least four points (see \autoref{fig:prop-e}), a contradiction.

\begin{figure}
    \centering
    \includegraphics[width=0.6\linewidth]{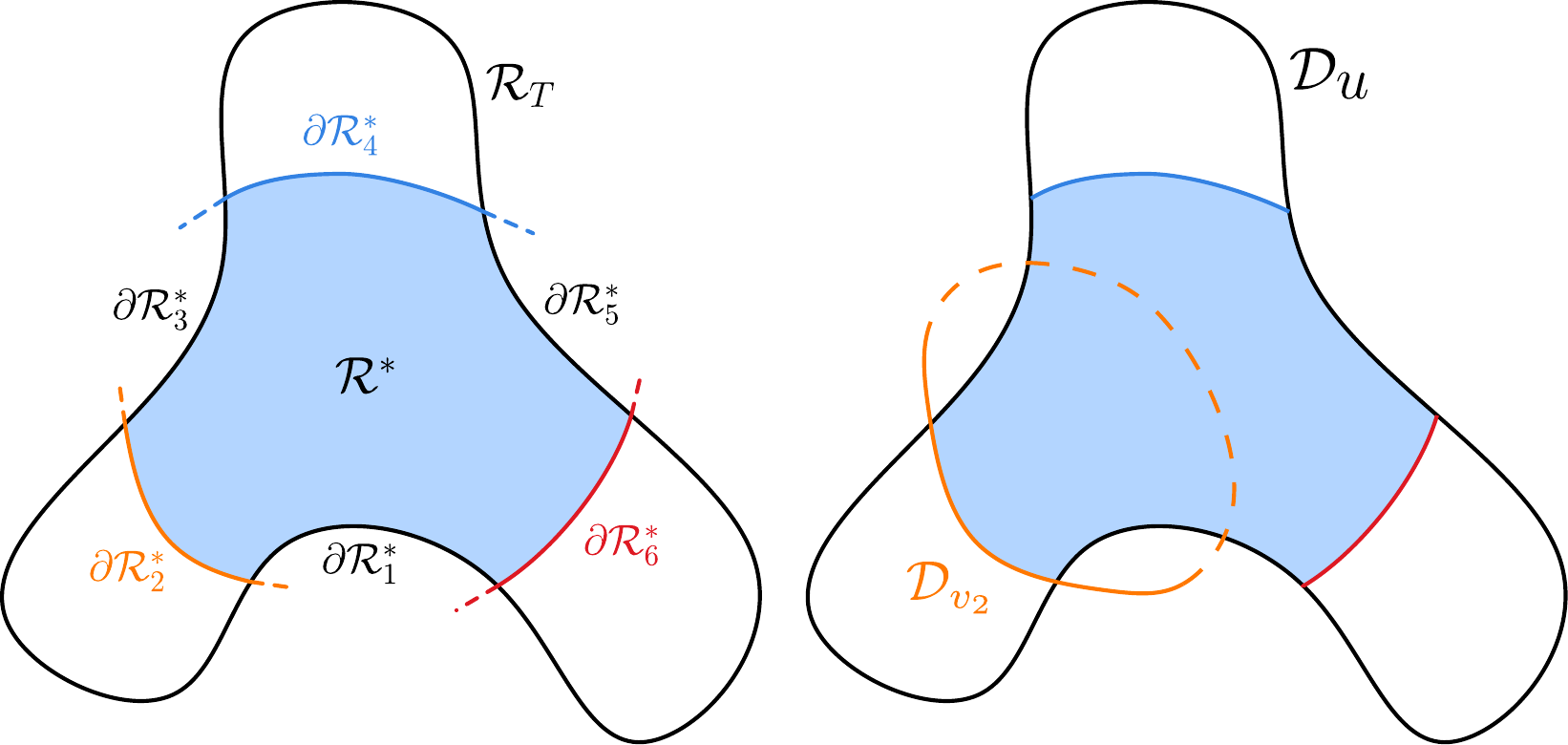}
    \caption{Proof of Property~\eqref{e:ee} (left) The region $\mR^*$. (right) The pseudo-disks $\D_u$ and $\D_{v_2}$.}
    \label{fig:prop-e}
\end{figure}

\cqed \qed \end{proof}

Let $\mS'$ be the restriction of $\mS$ to vertices of $M'$.
Let us denote $\PSprime$ the representation of $G[M']$, and $\dPSprime$ its dual. Since the pseudo-disk system has ply at most $p$, the graph $G[M']$ has $\O(p|M'|)=\O(p^3k)$ edges by \autoref{lem:pseudo-number-edges}.
Each of such edge induces at most two vertices in $\PSprime$. Hence, $\PSprime$ has $\O(p^3k)$ vertices, and thus by planarity, the number of edges, and faces in $\PSprime$ is also  $\O(p^3k)$.

\begin{definition}\label{def:Ht-intro}
Given a vertex $u\in V(G)\setminus M'$, we define the \emph{hosted graph} $H_u$ as a plane graph drawn within $\D_u$, with a (single) vertex in a face $f$ of $\PSprime$ if and only if $\D_u$ and $f$ intersect, and with edges between vertices lying in adjacent faces of $\PSprime$, if their common border intersects $\D_u$ (see \autoref{fig:M'+Hu}, right).
\end{definition}

\begin{claim}\label{cl:hosted-tree}
    For any vertex $u\in V(G)\setminus M'$, $H_u$ is a tree.
\end{claim}
\begin{proof}
 Let us first prove that $\D_u$ does not contain a crossing, meaning that there are no vertices $v_1,v_2\in V(G)$ such that $\D_u$ contains a point in $\pD_{v_1} \cap \pD_{v_2}$. Suppose by contradiction that this is the case. Clearly $uv_1v_2$ is a triangle so, by definition of $M'$, $\{v_1,v_2\} \subseteq M'$. But this then contradicts \autoref*{cl:pseudo-disk-abc}.\eqref{e:bb}.

 Now, as $\D_u$ does not contain a crossing, the family $\{\D_u \cap \D_v,\ v \in N(u)\}$ is a laminar set family. This implies that $H_u$ is a tree (where in particular the leaves of $H_u$ correspond to regions $\D_u \cap \D_v$, where $\D_u \cap \D_v$ is inclusion-wise minimal among the laminar family). \cqed \qed
\end{proof}

We can now proceed with the kernelization, that divides into two parts. The first part  (see \autoref{sssec:inM'}) deals with the inner-$M'$-vertices, while the second part deals with border- and outer-$M'$-vertices.
The second part is split into three sections.
In \autoref{sssec:notinM'intro}, we will study the properties of border-$M'$-vertices.
In \autoref{sssec:notinM'rules}, we will provide kernelizations rules, and we will prove that the kernel is equivalent to $(G,k)$ with respect to the \FVS problem. Then in \autoref{sssec:notinM'size}, we will prove the properties of the kernel $(G',k')$ stated in \autoref{lemma:red-fvs-pseudo}, in particular its small size.

\subsection{Inner-\texorpdfstring{$M'$}{M'}-vertices}\label{sssec:inM'}

The purpose of this section is to reduce the number of inner-$M'$-vertices. We distinguish two cases according to the maximum degree of $H_u$.
We begin with the vertices $u\in I_{M'}$ such that $\Delta(H_u)\ge 3$. Actually, for these vertices, we do not need to reduce their number. Indeed, we can bound their number by using the following consequence of Euler's formula.

\begin{lemma}[{see \cite[Lemma~9.24]{Cygan2015Book}}]
\label{lem:planar-linear-size}
    For any bipartite planar graph $G=(A,B,E)$, if $d(v)\ge 3$ for every $v\in A$, then we have that $|A| \leq 2|B|$.
\end{lemma}

\begin{claim}\label{cl:d3Ht-intro}
    The number of vertices $u\in I_{M'}$ such that $\Delta(H_u)\ge 3$ is $\O(p^3k)$.
\end{claim}
\begin{proof}
    To see this, let us construct a bipartite pseudo-disk graph as follows. In one part, we consider the arcs of $\PSprime$. These arcs are slightly shortened (so that they form an independent set) and slightly thickened (so that they intersect the adjacent regions we define below), see \autoref{fig:pseudo-bip}. In the other part, we define a pseudo-disk $\D'_u$ for each vertex $u\in I_{M'}$ such that $\Delta(H_u)\ge 3$ as follows. Consider a vertex $x$ of $H_u$ such that $d_{H_u}(x)\ge 3$, and let $\D'_u$ be the intersection of $\D_u$ and the face of $\PSprime$ containing $x$.
    Since triangle-free pseudo-disk graphs are planar \cite{kratochvil1996intersection}, by \autoref{lem:planar-linear-size}, the second part has size $\O(|E(\PSprime)|) = \O(p|M'|)= \O(p^3k)$. \cqed \qed
\end{proof}

\begin{figure}
    \centering
    \includegraphics[scale=0.3]{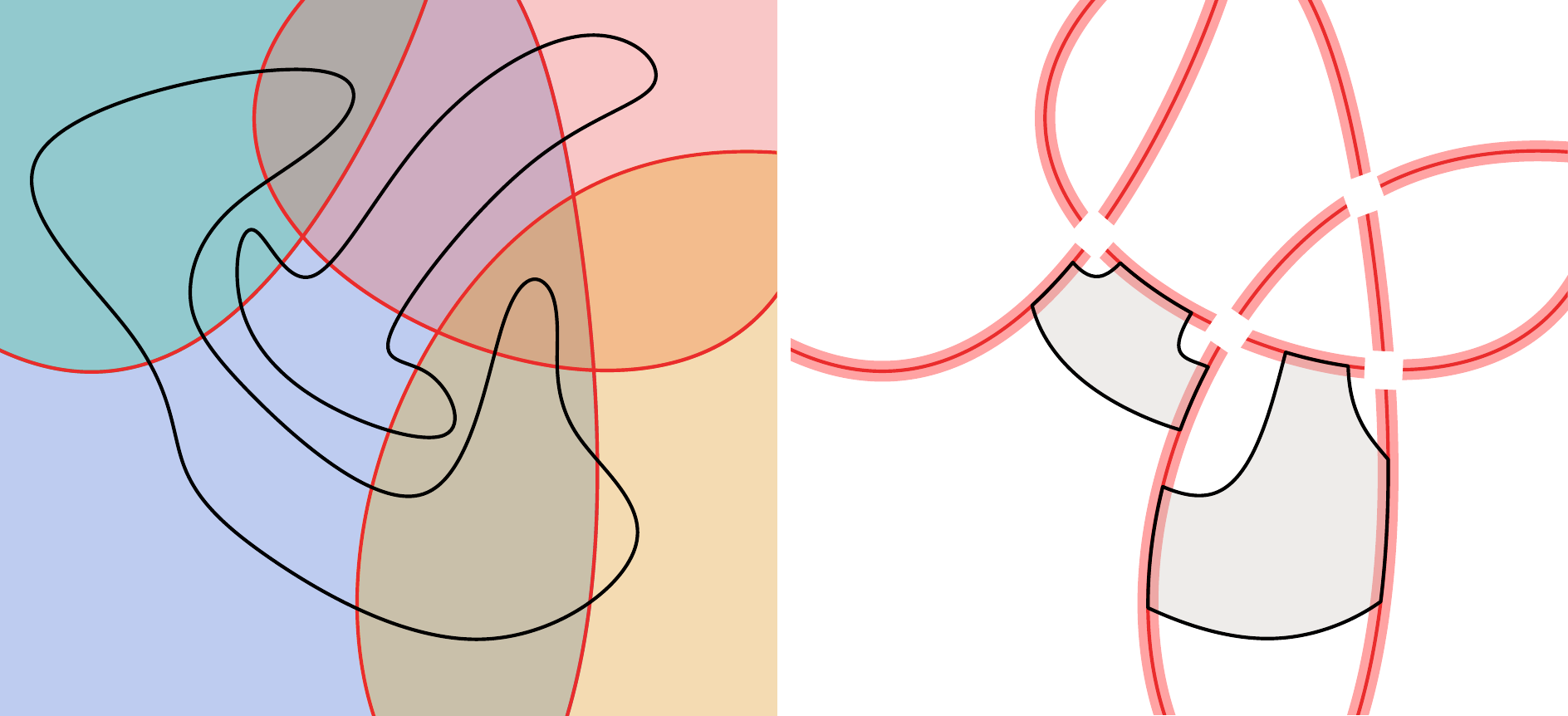}
    \caption{Two inner-$M'$ pseudo-disks with $\Delta(H_u)\ge 3$, and the bipartite pseudo-disk graph constructed to prove \autoref{cl:d3Ht-intro}.}
    \label{fig:pseudo-bip}
\end{figure}

\begin{figure}
    \centering
    \includegraphics[width=0.5\textwidth]{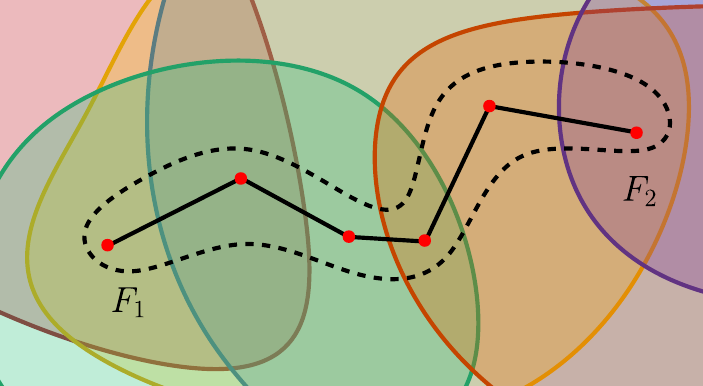}
    \hspace{.5cm}
    \includegraphics[width=0.45\textwidth]{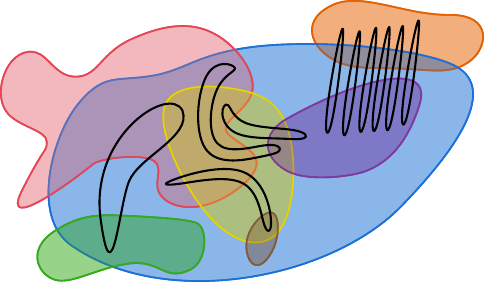}
    
    \caption{
    Left: a hosted graph $H_u$ that is a path. Right: some vertices $u \in I_{M'}$ such that $H_u$ is a path. Some of these vertices can have many twins.}
    \label{fig:pathsinside}
\end{figure}

We can now focus on vertices $u\in I_{M'}$ such that $\Delta(H_u)\le 2$ (see \autoref{fig:pathsinside}). In that case, $H_u$ is a path (with possibly only one vertex) of bounded length, as we show now.
\begin{claim}\label{claim:inner-path}
    For any vertex $u\in I_{M'}$ such that $H_u$ is a path, this path has length at most $d(u)$, and $d(u)\le 2p$. Furthermore, $N(u)$ can be split into two cliques of size at most $p$ each.
\end{claim}
\begin{proof}
    For any neighbor $v$ of $u$, $\D_u$ intersects $\pD_v$ at most once as otherwise $\pD_v$ and $\pD_u$ would intersect on more than two points which is not allowed in systems of pseudo-disks.
    Hence the number of edges in $H_u$ (i.e. the length of $H_u$) is at most $d(u)$.
    
    Furthermore, since $\Delta(H_u) \le 2$, we have that any neighbor $v$ of $u$ is such that $\D_v$ contains one of the faces $F_1$ and $F_2$ of $\PSprime$ (see the left of \autoref{fig:pathsinside}) corresponding to the endpoints of the path $H_u$ in $\dPSprime$. Since the representation has ply at most $p$, there are at most $2p$ pseudo-disks containing one (or two) of these faces. Finally, to split the neighborhood of $u$ into two cliques, it suffices to consider the pseudo-disks containing $F_1$, which clearly induce a clique, and the remaining ones, since all of them contain $F_2$.  \cqed \qed
\end{proof}

To bound the number of such vertices $u$ with $\Delta(H_u)\le 2$, we need to use the following kernelization rule for the \FVS problem.

\begin{enumerate}[(R1)]
\setcounter{enumi}{-1}
    \item \label{e:r0pd}  Consider two disjoint and possibly empty subsets $A,B\subseteq M'$ inducing complete graphs, and six isolated vertices of $G - M'$, $v_1,\ldots,v_5$ and $u$. If $N_G(u) \subseteq N_G(v_i) \subseteq A\cup B$, for $1\le i\le 5$, delete vertex $u$ and keep the same parameter $k$.
\end{enumerate}
\begin{proof}
    (Safeness of the rule) Let $G' = G-u$, and note that it cannot have a minimum feedback vertex set larger than $G$. Hence, it is sufficient to show that $G'$ has a minimum feedback vertex set $X$, that is also a feedback vertex set for $G$. 

    Let $A_u = A\cap N_G(u)$ and $B_u = B\cap N_G(u)$. Consider any minimum feedback vertex set $X$ of $G'$, and let us transform it (if needed) into a minimum feedback vertex set including all the vertices of $A_u\cup B_u$, except possibly one. Observe first that $G[A]-X$ and $G[B]-X$ being acyclic, we have that $|A \setminus X| \le 2$ and $|B \setminus X| \le 2$.
    
    If $|A_u\setminus X| = 2$, then $\{v_1,\ldots,v_5\} \subseteq X$, as otherwise the two vertices in $A_u\setminus X$ and any $v_i$ not in $X$ would form a triangle in $G' - X$. In such case, defining $X'$ by replacing in $X$ the vertices $v_1,\ldots,v_5$ with the vertices in $A\setminus X$ and in $B\setminus X$ would result in the desired feedback vertex set. Indeed, the vertices $v_1,\ldots,v_5$ are isolated in $G'-X'$, and $G'-\{v_1,\ldots,v_5\}-X'$ is a forest, as $X \subseteq X' \cup \{v_1,\ldots,v_5\}$.
    
    Hence, we consider that $|A_u\setminus X|\le 1$ and $|B_u\setminus X|\le 1$. If $|A_u\setminus X| = 1$ and $|B_u\setminus X|= 1$, then $|\{v_1,\ldots,v_5\} \cap X|\ge 4$, as otherwise the two vertices in $\{v_1,\ldots,v_5\} \setminus X$, the vertex in $A_u\setminus X$ and the vertex in $B_u\setminus X$ would form a 4-cycle in $G' - X$. In such case, defining $X'$ by replacing in $X$ four vertices in $\{v_1,\ldots,v_5\} \cap X$ with the vertices in $A\setminus X$ and in $B\setminus X$ would result in the desired feedback vertex set.
    Indeed, the vertices $v_1,\ldots,v_5$ not in $X'$ are isolated in $G'-X'$, and $G'-\{v_1,\ldots,v_5\}-X'$ is a forest, as $X \subseteq X'\cup \{v_1,\ldots,v_5\}$.

    We thus have a minimum feedback vertex set $X$ of $G'$ including all the vertices of $A_u\cup B_u$, except possibly one. Since $u$ has degree zero or one in $G - X$, we have that $G - X$ is the forest $G' - X$ with a new vertex $u$, that is either a leaf or an isolated vertex, and it is thus a forest. \cqed \qed
\end{proof}

Identifying the configuration required by the previous rule and updating the representation of $G$ when we delete vertices can be done in polynomial time (see \autoref{lem:tech-pseudo-disk-del-contr}), implying that this kernelization can be performed in time polynomial in $|V(G)|$.

Note that from now on, among the vertices $u\in I_{M'}$ such that $H_u$ is a path (as their neighborhoods induce two complete graphs by \autoref{claim:inner-path}), there are at most five vertices with the same neighborhood. We can thus focus on bounding the number of such vertices $u$ with distinct neighborhoods. Since these vertices have bounded degree (by \autoref{claim:inner-path} again), the following theorem provides us a bound on their number.

\begin{theorem}[\cite{neighpseudo}]\label{hyperpseudo}
    Suppose $\EuScript{F}$ is a family of pseudo-disks in the plane and $\EuScript P$ is a finite subset of $\EuScript F$. Consider the hypergraph $H(\EuScript P,\EuScript F)$ whose vertices 
are the pseudo-disks in $\EuScript P$ and the edges are all subsets of~$\EuScript P$ of the
form $\{D \in \EuScript P \mid D \cap S \neq \emptyset\}$, where $S$ is a pseudo-disk in $\EuScript F$.
  Then the number of edges of cardinality at most $k\geq 1$ in $H(\EuScript P,\EuScript F)$ is $\O(|\EuScript P|k^3)$.
\end{theorem}

This allows us to bound the number of inner-$M'$-vertices by $\O(|M'|p^3)$, but the following claim allows a small improvement.
Let us first note that every inner-$M'$-vertex was already an inner-$M$-vertex.
\begin{claim}\label{claim:inner-M-already}
    For every vertex $u\in I_{M'}$, we have that $\D_u$ is contained in $\RM = \left( \cup_{v \in M} \D_v\right)$.
\end{claim}
\begin{proof}
Suppose, towards a contradiction, that for some vertex $u\in I_{M'}$, we have $\D_u \not\subseteq \RM$. Let $x$ be a point of $\D_u \setminus \RM$ arbitrarily close to $\RM$. As $\D_u \subseteq \RMprime$, there is a vertex $w\in M'\setminus M$ such that $x\in \D_w$. the point $x$ being arbitrarily close to $\RM$, there is a vertex $v \in M$ such that $\D_u$, $\D_v$, and $\D_w$ intersect. This contradicts the fact that every triangle of $G$ has at least two vertices in $M$.     \cqed \qed
\end{proof}

A vertex $u\in I_{M'}$ such that $H_u$ is a path has degree at most $2p$ (by \autoref{claim:inner-path}), and all of them were in $M$ (by \autoref{claim:inner-M-already}). Hence, \autoref{hyperpseudo} implies that there are at most $\O(|M|p^3)$ such vertices with distinct neighborhoods.
Thus, together with \autoref{cl:d3Ht-intro} we have the following bound.
\begin{claim}\label{cl:Dv-inside-bounded}
    After the kernelization of \autoref{sssec:inM'}, the number of inner-$M'$-vertices is $\O(p^4k)$.
\end{claim}

\subsection{Properties of Border-\texorpdfstring{$M'$}{M'}-vertices}
\label{sssec:notinM'intro}

With $\mS$ fixed, we define the \emph{$M'$-ply} of a point of the plane as the number of pseudo-disks from $M'$ that contain it.
We consider the border-$M'$-vertices $v$, that is, those such that $\D_v$ intersects $\RMprime$, without being included in it.
For such vertices, we can strengthen \autoref{cl:hosted-tree}, but we have to orient the edges of the hosted graphs $H_u$.
The \emph{directed hosted graph} $\oH{u}$ is an orientation of $H_u$
such that the edge $ab$ is oriented from the vertex with lower $M'$-ply towards the other one. This is always possible as in a pseudo-disk representation, any two incident faces have different ply.

\begin{figure}[!ht]
    \centering
    \includegraphics[width=\textwidth]{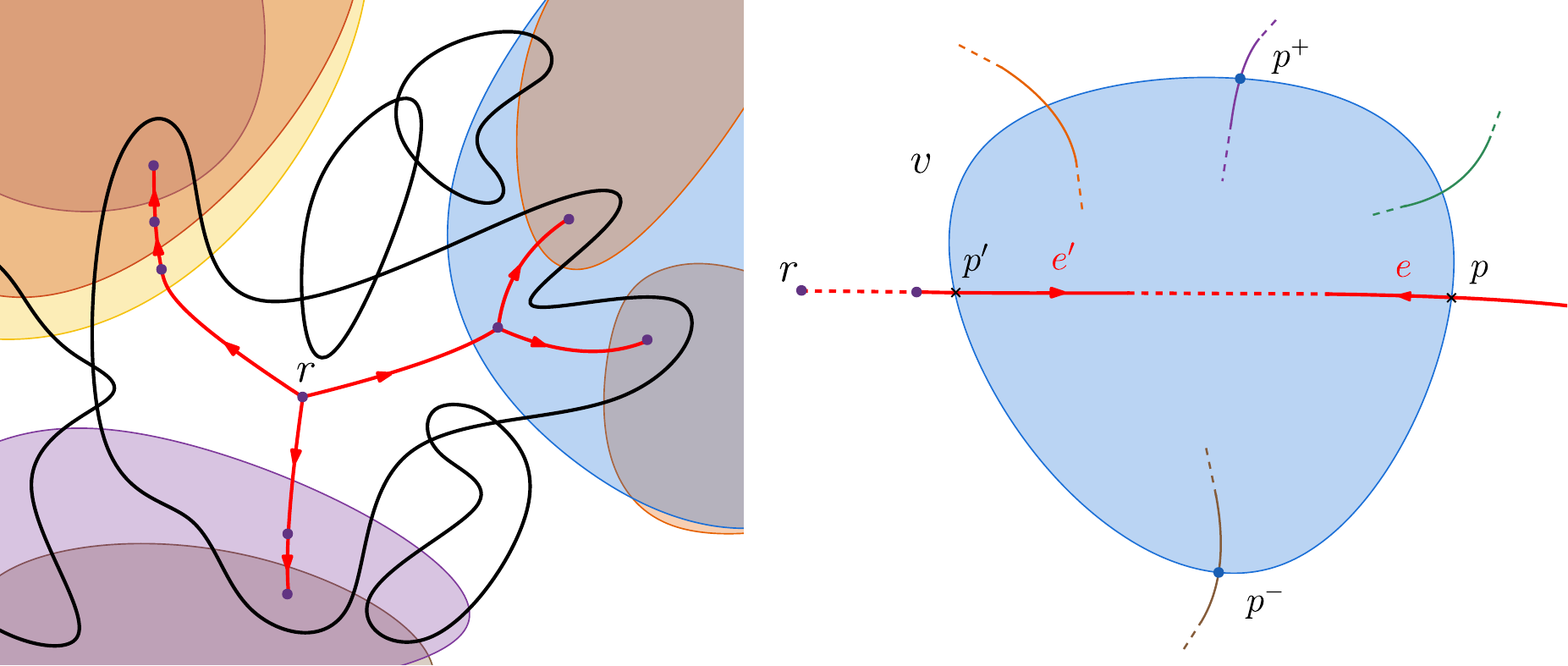}
    \caption{(left) The filled pseudo-disks belong to $M'$, and the empty one correspond to a tree of $G-M'$ with vertices in $B_{M'}\cup O_{M'}$. For one of these vertices $u\in B_{M'}$, the graph $\oH{u}$ is drawn inside $\D_u$. (right) Proof of \autoref{cl:hosted-BM-arbo}.}
    \label{fig:Hu-arbo}
\end{figure}

An \emph{arborescence} is an oriented tree with exactly one source, called the root.
\begin{claim}\label{cl:hosted-BM-arbo}
    For any vertex $u\in B_{M'}$, the directed hosted graph $\oH{u}$ is an arborescence whose root is drawn in a face of $\PSprime$ of $M'$-ply zero (see an example in the left of \autoref{fig:Hu-arbo}).
\end{claim}

\begin{proof}
    One consequence of \autoref*{cl:pseudo-disk-abc}\eqref{e:ee} is that $\oH{u}$ has exactly one vertex corresponding to a $M'$-ply zero face in  $\PSprime$, call it the root $r$.
    Towards a contradiction, assume there is an arc $e$ pointing in the wrong direction, that is, towards $r$. Since $r$ has $M'$-ply zero, the pseudo-disk $\D_v$ whose border is crossed by $e$, is also crossed in the path linking $r$ and $e$ in $\oH{u}$ (see the right of \autoref{fig:Hu-arbo}). Denote $p$ and $p'$ the crossing point between $\pD_v$ and the arcs $e$ and $e'$, respectively. Since $e$ and $e'$ are distinct arcs, the two parts of $\pD_v\setminus \{p,p'\}$ are crossed by other borders of $M'$, at $p^+$ and $p^-$ say. Since $\D_u$ cannot contain these points $p^+$ and $p^-$ (by definition of $M'$), but contains $p$ and $p'$, we have that $\pD_u$ and $\pD_v$ intersect at least four times (in each of the arcs of $\pD_v\setminus \{p,p^-,p',p^+\}$), a contradiction. \cqed \qed
\end{proof}

We have to define new plane graphs. For any vertex $u\in B_{M'}$, the digraph $\oHp{u}$ is a plane arborescence drawn in the vicinity of $\oH{u}$, and it is such that 
\begin{itemize}
    \item $\oHp{u}$ has the same root and same leaves as $\oH{u}$, but
    \item it is a subdivision of an out-star whose paths from $r$ to any leaf $\ell$ have the same length as in $\oH{u}$. 
\end{itemize}
This graph can have several vertices in the same face of $\PS$, but each arc intersects a single border of $M'$, at a single point in the vicinity of $\oH{u}$ (see the left and the middle of \autoref{fig:oH-oHp-oHpp}).
\begin{figure}
    \centering
    \includegraphics[width=\textwidth]{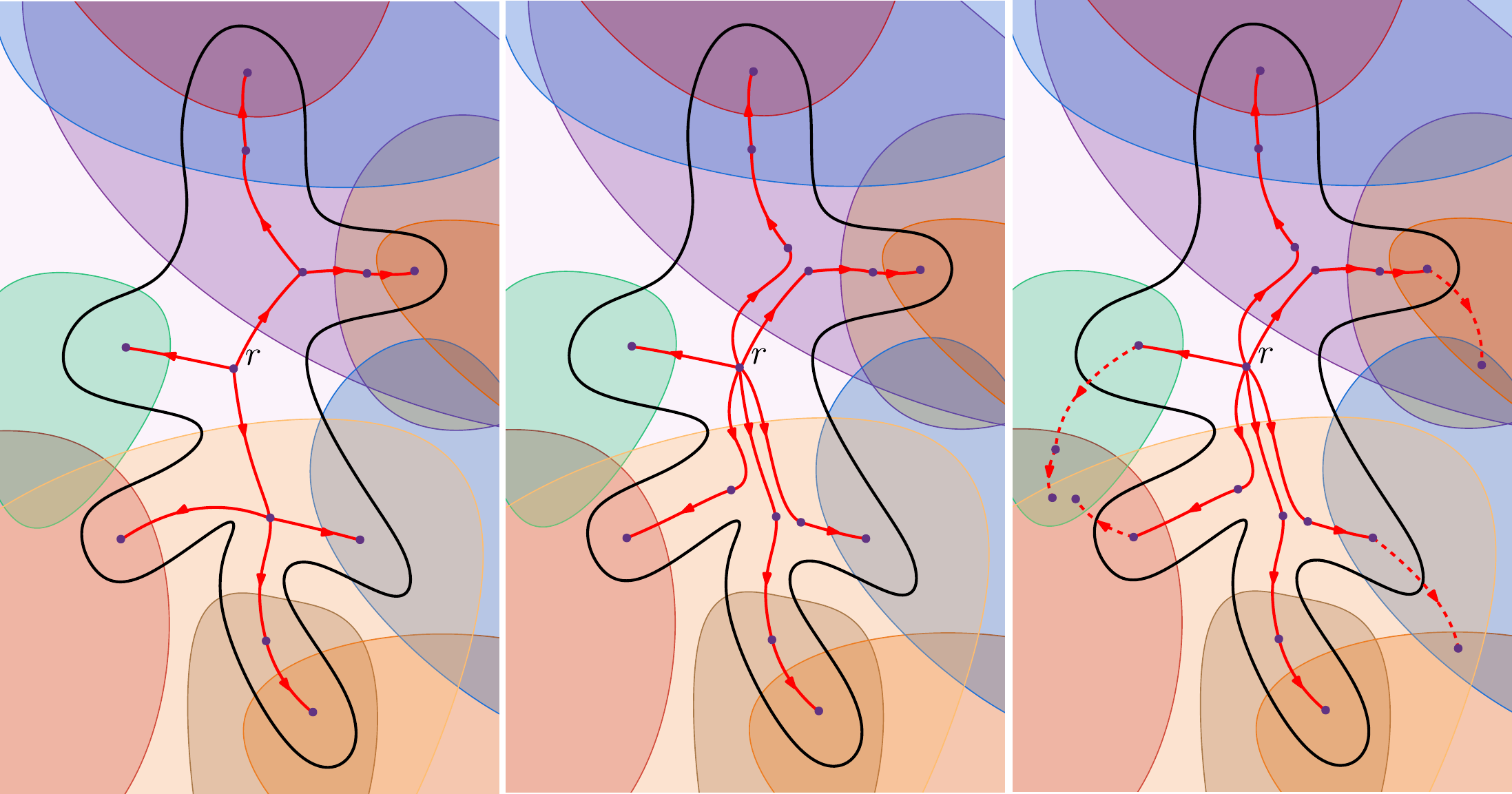}
    \caption{Example of graphs $\oH{u}$, $\oHp{u}$, and $\oHpp{u}$.}
    \label{fig:oH-oHp-oHpp}
\end{figure}

We now have to prolong these graphs $\oHp{u}$.
\begin{claim}\label{cl:oHpp}
    In polynomial time, the drawn arborescences $\oHp{u}$, for all $u\in B_{M'}$, can be prolonged into arborescences $\oHpp{u}$ fulfilling the following properties (see an example in the right of \autoref{fig:oH-oHp-oHpp}).
\begin{itemize}
    \item The root of $\oHpp{u}$ is the same as $\oHp{u}$.
    \item $\oHpp{u}$ is the subdivision of an out star.
    \item Each leaf of $\oHpp{u}$ is drawn in a face of $\PSprime$ whose $M'$-ply is locally maximal (i.e. such that all the incident faces have lower $M'$-ply).
    \item Each edge of $\oHpp{u}$ crosses exactly one edge of $\PSprime$, at a single point, and it is oriented from the smaller $M'$-ply towards the other end.
    \item When considering all the drawings $\oHpp{u}$ for all the $u\in B_{M'}$ there are no intersections. 
\end{itemize}
\end{claim}
\begin{proof}
All these properties, except the third one, are fulfilled by the original $\oHp{u}$. Thus, we only have to show how to prolong an arborescence if one of its leaves $\ell$ is not in a face with locally maximal $M'$-ply. Around the face $f$ of $\PSprime$ containing $\ell$, there is an edge $e$  of $\PSprime$, incident to a face $f'$ with larger $M'$-ply. If it is possible without intersecting an arc of a $\oHp{v}$, for some $v\in B_{M'}$, draw an arc towards $f'$ crossing only $e$ (see the left of \autoref{fig:wideH}). If it is not possible, then our drawing was blocked by some $\oHp{v}$ (possibly $v=u$), and in that case we just have to follow the arcs of $\oHp{v}$ towards the next face $f''$ of $\PSprime$ (see the right of \autoref{fig:wideH}).    \cqed\qed
\end{proof}

\begin{figure}
    \centering
    \includegraphics[scale=0.8]{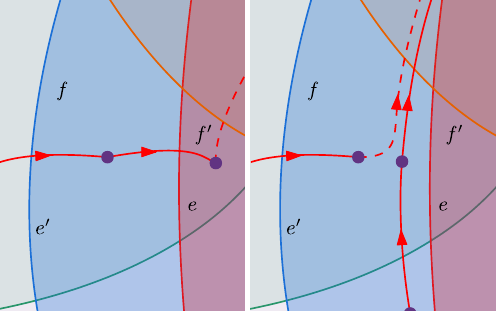}
    \caption{In red, the graph $\oHp{u}$. Left: the case where $e$ is reachable. Right: the case where $e$ is not reachable.}
    \label{fig:wideH}
\end{figure}

\begin{definition}
    Given a vertex $u\in B_{M'}$, the corresponding $\oHpp{u}$, and a root-to-leaf path $P''$ of $\oHpp{u}$, we denote $\type(P'')$ the directed path in $\dPSprime$ corresponding to $P''$. We say that $\type(P'')$ is the \emph{type} of $P''$.
    Then, $\types(u) = \{\type(P'')\ |\ P'' \text{ is a root-to-leaf path of } \oHpp{u}\}$. Given a tree $T$ of $G-M'$, with $V(T) \subseteq B_{M'}\cup O_{M'}$, the set $\types(T)$ is defined by $\types(T) = \bigcup_{u\in V(T)\cap B_{M'}} \types(u)$. The cardinalities of these sets are denoted $\numtypes(u)$ and $\numtypes(T)$. Finally, let $\alltypes$ be the set $\alltypes = \bigcup_{u\in B_{M'}} \types(u)$. 
\end{definition}

Note that by construction of $\oHpp{u}$, $\numtypes(u)$ corresponds to the number of leaves in $\oH{u}$.
Any face $f$ of $\PSprime$ with locally maximal $M'$-ply, only contains leaves of the graphs $\oHpp{u}$, for $u\in B_{M'}$.
When going around the border of such face $f$, one crosses, with a given order, the ending arcs of several root-to-leaf paths of some graphs $\oHpp{u}$. This order is strongly related with the type of these root-to-leaf paths.
\begin{claim}
    \label{cl:same-type-consecutive}
    Given a face $f$ of $\PSprime$ with locally maximal $M'$-ply, and given any path $P\in \alltypes$ ending in $f$, all the root-to-leaf paths of type $P$ end at $f$, and these paths are crossed consecutively when going around $f$. In other words, if going around $f$ one crosses, in this order, four root-to-leaf paths $P''_1,P''_2,P''_3,P''_4$ of some graphs $\oHpp{u}$, for $u\in B_{M'}$ (these paths can belong to one, two, three or four different graphs $\oHpp{u}$),
    we cannot have that $\type(P''_1)=P_a$, $\type(P''_2)=P_b$, $\type(P''_3)=P_a$, and $\type(P''_4)=P_c$, for any triple $P_a,P_b,P_c \in \alltypes$ such that $P_a\neq P_b$ and $P_a\neq P_c$.
\end{claim}
\begin{proof}
    Note that if $\type(P''_1) = \type(P''_3)$, then $P''_1$ and $P''_3$ go through the same faces of $\PSprime$, in the same order, and only visit one $M'$-ply zero face, the starting one.
    Hence, one of $P''_2$ or $P''_4$ has to visit exactly the same faces, and is thus of the same type (see \autoref{fig:pairs-type}).\cqed \qed
\end{proof}

\begin{figure}[!ht]
    \centering
    \includegraphics[width=\textwidth]{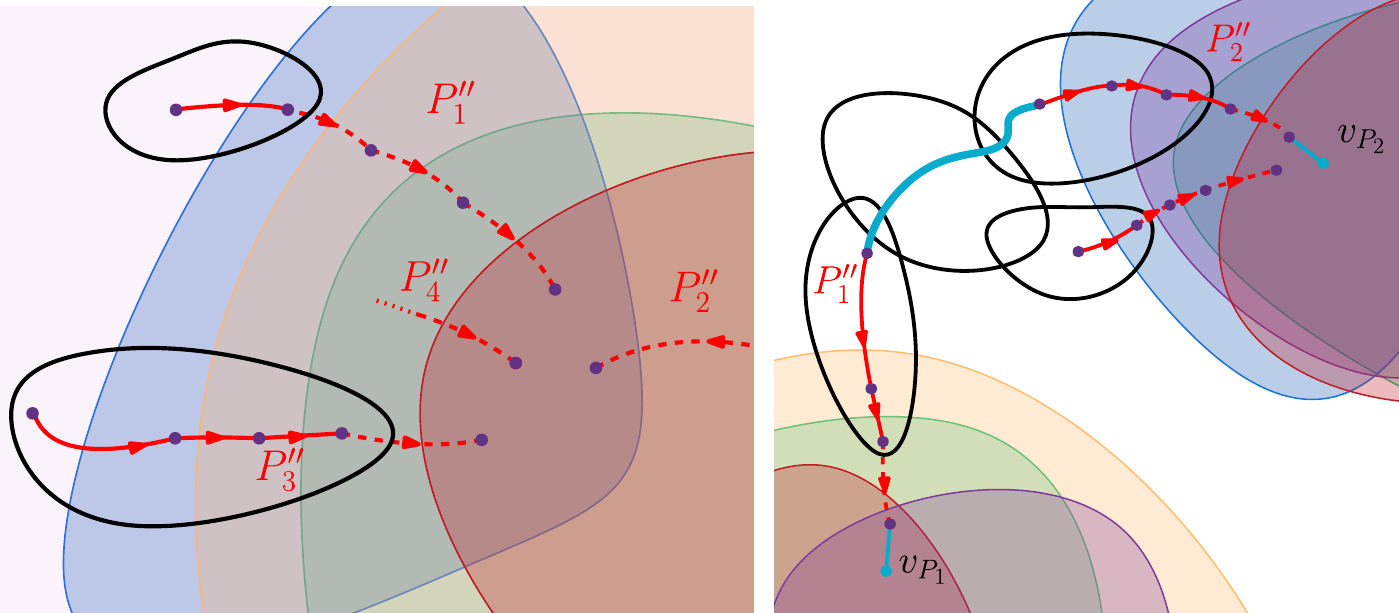}
    \caption{Illustration for the proofs of \autoref{cl:same-type-consecutive} and \autoref{cl:few-degree-two-T}.}
    \label{fig:pairs-type}
\end{figure}

We now show that the number of types is small.
\begin{claim}\label{cl:wV-small}
    The set $\alltypes$ has size $\O(p|M'|)$.
\end{claim}
\begin{proof}
    Consider the bipartite plane multigraph $B=(F_0,F_{\rm max}, \alltypes)$ defined as follows. The part $F_0$ corresponds to the faces of $\PSprime$ with $M'$-ply zero, the part $F_{\rm max}$ to those with locally maximal $M'$-ply, and there is an edge from $f\in F_0$ to $f'\in F_{\rm max}$ for each path $P$ of $\alltypes$ linking these two faces. These edges can be drawn by following one of the directed paths $P''$ such that $\type(P'')=P$. The planarity of $B$ follows from the fact that such paths $P''$ are not crossing.

    This graph has $\O(p|M'|)$ vertices (since those are faces of $\PSprime$), and $\O(p|M'|)$ faces of degree two (as any such face contains a vertex of $\PSprime$). By Euler's formula the number of edges is also $\O(p|M'|)$.\cqed \qed
\end{proof}  

We are now going to consider pairs of paths of $\alltypes$, which appear in the same set $\types(T)$, for some tree $T$ of $G-M' -I_{M'}$.
\begin{claim}\label{cl:few-degree-two-T}
    There are $\O(p|M'|)$ pairs $\{P_1,P_2\} \subseteq \alltypes$ such that $\types(T)=\{P_1,P_2\}$, for some tree $T$ of $G-M' -I_{M'}$.
\end{claim}
\begin{proof}
    We are again going to construct a planar graph (actually a plane graph).
    Its vertices $v_P$ are in bijection with the paths $P\in \alltypes$, and they are drawn in the face of $\PSprime$ corresponding to the high $M'$-ply endpoint of $P$. The vertices $v_P$ in the same face $f$ are drawn close to its border, in the same order around, as the one given by \autoref{cl:same-type-consecutive}. 

    Then we draw an edge from $v_{P_1}$ to $v_{P_2}$ for each pair $\{P_1,P_2\} \subseteq \alltypes$ such that $\types(T)=\{P_1,P_2\}$, for some tree $T$ of $G-M' -I_{M'}$. To draw such an edge, consider a given tree $T$ of $G-M' -I_{M'}$ such that $\types(T)=\{P_1,P_2\}$. Then, consider paths $P''_1$, $P''_2$ of $\oHpp{u_1}$, $\oHpp{u_2}$ for some vertices $u_1$, $u_2$ of $T$ (possibly $u_1=u_2$) such that $\type(P''_1)=P_1$ and $\type(P''_2)=P_2$.
    The edge from $v_{P_1}$ to $v_{P_2}$ can then be drawn by following $P_1''$ and $P''_2$ and connecting them inside of the union of the pseudo-disks of $T$ (see \autoref{fig:pairs-type}). Since the considered trees are not adjacent (they are connected components of $F-M'-I_{M'}$), the unions of their respective pseudo-disks are pairwise disjoint hence the edges drawn as described above do not cross. This shows that the graph we construct is planar.
    Finally this graph being planar and simple its number of edges is linear in terms of its number of vertices. Hence, there are $\O(|\alltypes|)=\O(p|M'|)$ pairs $\{P_1,P_2\} \subseteq \alltypes$ such that $\types(T)=\{P_1,P_2\}$, for some tree $T$ of $G-M' -I_{M'}$.\cqed \qed
\end{proof}

In the next section, we will bound the number of trees $T$ of $G-M'-I_{M'}$ such that $\types(T)= \{P_1,P_2\}$ for some pair $\{P_1,P_2\}$, and we will bound their size. We will also have to deal with trees $T$ such that $\numtypes(T)\ge 3$. This is why we need the following property.
\begin{claim}\label{cl:few-large-degree-T}
\[
\sum_{\substack{T \in G-M' -I_{M'}\\ \text{ with}\numtypes(T)\ge 3}} \numtypes(T) = \O(p|M'|).
\]
\end{claim}
\begin{proof}
    Again, this is shown by constructing a planar graph. Here, it will be simple and bipartite, with one part containing one vertex $v_T$ per tree $T \in G-M' -I_{M'}$ such that $\numtypes(T)\ge 3$, and one part containing one vertex $v_P$ per path $P\in \alltypes$.
    The vertices $v_P$ are drawn as in the proof above, while the vertices $v_T$ are drawn in the $M'$-ply zero face of $\PSprime$ intersecting the pseudo-disks of $T$ (by \autoref*{cl:pseudo-disk-abc}\eqref{e:ee}).

    Then we draw an edge from $v_{T}$ to $v_{P}$ if and only if $P\in \types(T)$. To draw such an edge, consider a path $P''$ of $\oHpp{u}$ for some vertices $u$ of $T$ such that $\type(P'')=P$. Then starting from $v_P$ we can follow $P''$ and then continue inside of the union of the pseudo-disks of $T$ up to $v_T$, similarly to what we did in the proof of \autoref{cl:few-degree-two-T} (see \autoref{fig:deg3-type}).
    Finally, this graph being bipartite planar and simple, \autoref{lem:planar-linear-size} allows us to conclude.\cqed \qed
\end{proof}
\begin{figure}
    \centering
    \includegraphics[width=\textwidth]{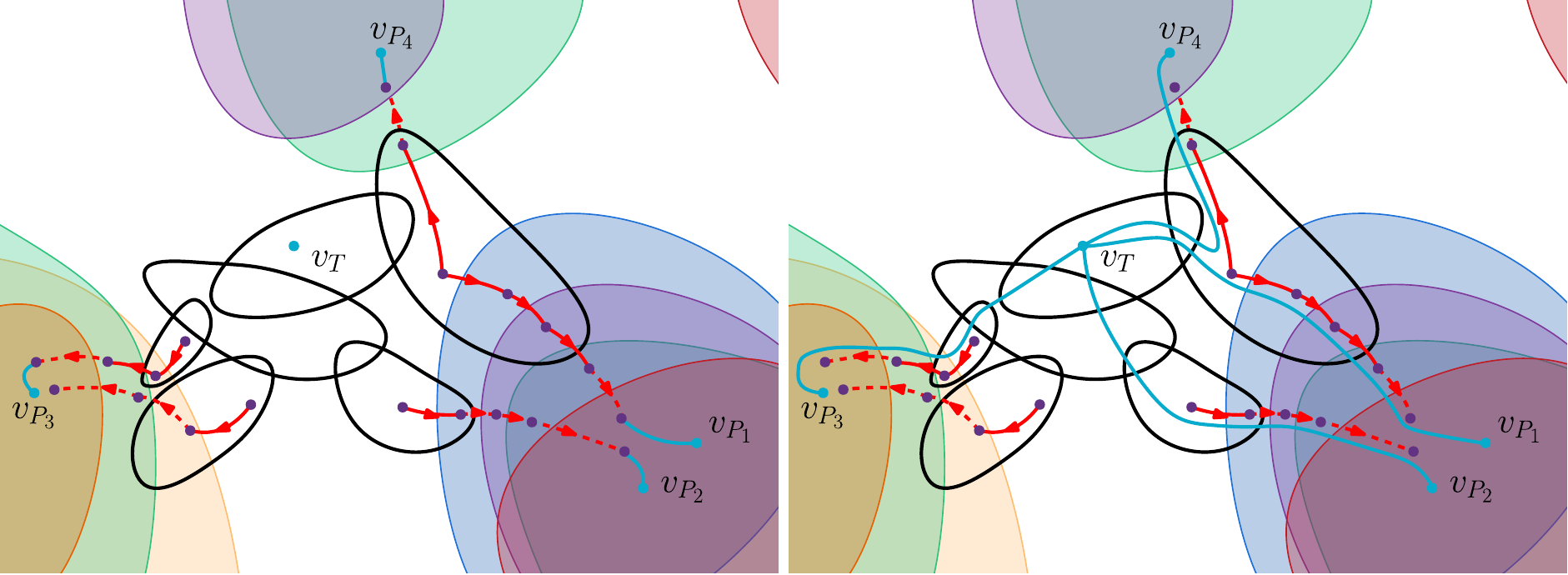}
    \caption{Illustration for the proof of \autoref{cl:few-large-degree-T}.}
    \label{fig:deg3-type}
\end{figure}

In the next section, we will bound the size of any tree $T$  of $G-M' -I_{M'}$ by $\O(\numtypes(T))$, and together with the previous claims, this will allow us to bound $|B_{M'}\cup O_{M'}|$.
Towards this goal, we have to introduce the following notions.

Given a path (type) $P\in \alltypes$, we denote by $K_P$ the $q$-tuple $(v_1,\ldots,v_q)$, where $v_i\in M'$ is the $i^{th}$ pseudo-disk entered by $P$, and thus $q$ is the $M'$-ply of the final vertex of $P$.

Note that since $P$ is a directed path and each arc goes from a region of lower $M'$-ply to a region of higher $M'$-ply, $\{v_1, \dots, v_q\}$ actually forms a clique, hence the name $K_P$. 
Let  $\ktypes{} = \{K_P\ |\ P\in \alltypes\}$.
\begin{definition}
\label{def:KmonoNeighborhood}
    Given a $q$-clique of $G[M']$, for some $q\le p$, and an ordering of its vertices $K = (v_1,\ldots,v_{q})$ (here $K$ is a $q$-tuple), a vertex $u$ has \emph{$K$-monotone neighborhood} if there is an integer $q_u$ such that:
    \[
    N(u)\cap M'= \{v_1,\ldots,v_{q_u}\}.
    \]
    In this context, $q_u$ is the \emph{$K$-index} of vertex $u$.  
    Given a $q$-clique and a $q'$-clique of $G[M']$ (possibly sharing some vertices), for some $q,q'\le p$, and two orderings of their vertices $K = (v_1,\ldots,v_{q})$ and $K' = (v'_1,\ldots,v'_{q'})$, a vertex $u$ has \emph{$(K,K')$-monotone neighborhood} if there exist two integers $q_u,q'_u$ such that:
    \[
    N(u)\cap M'= \{v_1,\ldots,v_{q_u}\}\cup \{v'_1,\ldots,v'_{q'_{u}}\}.
    \]
    In this context, $q_u$ and $q'_u$ are the \emph{$K$-index} and the \emph{$K'$-index} of vertex $u$, respectively.
\end{definition}

Let us see how $\numtypes(u)$ relates to monotone neighborhoods.
\begin{claim}\label{cl:wd1}
    Given a vertex $u\in B_{M'}$ such that $\numtypes(u)=1$, and given $K=(v_1,\ldots, v_q)$ the ordered list of the borders of $M'$ crossed by the unique path $P\in \types(u)$. Then the vertex $u$ has a $K$-monotone neighborhood.
\end{claim}
\begin{proof}
    This clearly follows from the fact that $\oHpp{u}$ is constructed by prolonging the unique path of $\oHp{u}$, and the fact that $N(u)\cap M'$ is exactly the set of borders crossed by this path of $\oHp{u}$.
    Note that if the last border crossed by $\oHp{u}$ is the border of $v_i$, then the $K$-index of $u$, $q_u$, is $i$.
    \cqed \qed
\end{proof}
\begin{claim}\label{cl:wd2}
    Given a vertex $u\in B_{M'}$ such that $\numtypes(u)=2$, and given the ordered lists of the borders of $M'$ crossed by the two paths in $\types(u)$, $K=(v_1,\ldots, v_q)$ and $K'=(v'_1,\ldots, v'_{q'})$. Then the vertex $u$ has a $(K,K')$-monotone neighborhood.
\end{claim}
\begin{proof}
    Again, this clearly follows from the fact that $\oHpp{u}$ is constructed by prolonging the two paths of $\oHp{u}$, and the fact that $N(u)\cap M'$ is exactly the set of borders crossed by these paths of $\oHp{u}$.
    Again, the last borders crossed by the branches of $\oHp{u}$ are those of $v_{q_u}$ and $v'_{q'_u}$, where $q_u,q'_u$ are the $K$- and $K'$-indices of $u$.\cqed \qed
\end{proof}

\subsection{Kernelization rules}\label{sssec:notinM'rules}

Here, we provide kernelization rules for \FVS that will allow us to bound the size of the trees in $G - M' - I_{M'}$. Most of these rules are related to the paths in $\alltypes$ and to the $K$- and $(K,K')$-monotone-neighborhood property satisfied by the vertices $u\in B_{M'}$, such that $\numtypes(u)=1$ or 2, respectively. As the kernelization has to be performed in polynomial time (in terms of $n=|V(G)|$), we observe that within such a time, one can compute the graphs $\oH{u}$, $\oHp{u}$, and  $\oHpp{u}$, and the sets $\type(P'')$, $\types(u)$, $\types(T)$, $\alltypes$, $\ktypes{}$. Furthermore, note that updating these sets when a vertex $v\in M'$ or $u\in V(G)\setminus M'$ is deleted, can also be performed in polynomial time (by \autoref{lem:tech-pseudo-disk-del-contr}). Finally, note that such update is also possible when contracting an edge $uu'$ of $G-M'$ (by \autoref{lem:tech-pseudo-disk-del-contr} also). 

Let us consider the following kernelization rules for \FVS. Those are applied prioritizing the first rules. As those are folklore, we omit the proofs of the two first rules.
\begin{enumerate}[(R1)]
    \item\label{e:r1pd} Delete every vertex $v$ with degree at most one, and maintain the parameter $k$.
    \item\label{e:r2pd} Consider two adjacent vertices $u,u'\notin M'$, such that $d(u)=2$, and such that $u$ and $u'$ do not have a common neighbor. Then, contract the edge $uu'$, and maintain the parameter $k$.
\end{enumerate}      

It should be noted that after applying \ruleref{e:r2pd}, the property that every triangle has at least two vertices in $M'$ might be lost.
Note that the rule \ruleref{e:r2pd} implies that from now on, any degree two vertex $u\notin M'$ is in a triangle, or has its two neighbors in $M'$. Hence in particular, any degree two vertex $u\notin M'$ has at least one neighbor in $M'$.
The following five rules deal with vertices having $K$-monotone neighborhoods with $K\in \ktypes{}$.

\begin{enumerate}[(R1)]
\setcounter{enumi}{2}
    \item \label{e:r3pd} Consider a $q$-tuple $K=(v_1,\ldots,v_{q})\in \ktypes{}$, and two isolated vertices $u,w$ of $G - M'$ having $K$-monotone neighborhoods. If their $K$-indices are such that $1\le q_u\le q_w$, delete $u$ and maintain the parameter $k$.
\end{enumerate}

\begin{proof} (Safeness of \ruleref{e:r3pd})
Let us prove that if $G - u$ has a feedback vertex set of size $k$, so does $G$. Consider a feedback vertex set $X$. Recall that any clique has all its vertices in $X$, except at most two of them.
If $w\in X$, replace it with some vertex of $K\setminus X$, if any. Now we have that there is at most one vertex in $N(w)\setminus X$ (and thus in $N(u)\setminus X$). In other words,
$u$ and $w$ have degree at most one in $G - X$. Thus, since $(G-\{u,w\}) - X$ is a forest (as $X\cup\{w\}$ is a feedback vertex set), so is $G - X$.\cqed \qed
\end{proof}

\begin{enumerate}[(R1)]
\setcounter{enumi}{3}
    \item \label{e:r4pd} Consider a $q$-tuple $K=(v_1,\ldots,v_{q})\in \ktypes{}$, and a  tree $T$ of $G - M'$ whose vertices have $K$-monotone neighborhoods. If $T$ has at least two vertices, delete $v_1$ and decrease the parameter $k$ by one.
\end{enumerate}
\begin{proof} (Safeness of \ruleref{e:r4pd}) 
As $T$ has at least two vertices, it has at least two leaves, say $a$ and $b$.
Since the rule \ruleref{e:r1pd} does not apply for $a$ or $b$, these vertices are both adjacent to $v_1$. Consider a feedback vertex set $X$ not containing $v_1$.
Since there is at least one cycle in $\{v_1\}\cup V(T)$, this feedback vertex set intersects $V(T)$ on at least one vertex, denoted $x$. Let us prove that replacing $x$ with $v_1$ in $X$, we still have a feedback vertex set. Towards a contradiction, we consider a cycle $C$ not hit by this set. This cycle $C$ goes through $x$ and through some vertices of $\{v_1,\ldots,v_{q}\}$, since $T$ is acyclic.
However, since our new hitting set contains at least (and thus exactly) $q-1$ vertices among $\{v_1,\ldots,v_{q}\}$, the cycle has two vertices $a',b'\in V(T)$ adjacent to some vertex $v_i\in K\setminus X$, with $i>1$. Note that $X$ does not hit one of the triangles $a'v_1v_i$ or $b'v_1v_i$, as $a'\neq x$ or $b'\neq x$, a contradiction. We can thus  force $v_1$ in the feedback vertex set.\cqed \qed
\end{proof}

\begin{claim}\label{cl:few-wedg}
    After the kernelization described above, there are $\O(p|M'|)$ vertices in trees $T$ of $G-M'-I_{M'}$ such that $\numtypes(T)\le 1$.
\end{claim}
\begin{proof}
Rule \ruleref{e:r4pd} ensures us that every such tree has at most one vertex.
Rule \ruleref{e:r3pd} ensures that there is at most one such tree such that 
$\types(T)=\{P\}$, for any given $P\in \alltypes$.
So in total, there are at most $|\alltypes|=\O(p|M'|)$ such vertices (By \autoref{cl:wV-small}). \cqed \qed
\end{proof}

The following rules deal with the connected components of $G - M'$. This subgraph being acyclic, these connected components are trees. Note that the leaves in these trees have degree at least two in $G$ (by rule \ruleref{e:r1pd}), so they are necessarily adjacent to vertices in $M'$.

\begin{enumerate}[(R1)]
\setcounter{enumi}{4}
    \item \label{e:r5pd} Consider a $q$-tuple $K=(v_1,\ldots,v_{q})\in \ktypes{}$, and a tree $T$ of $G-M'$ (which is non-necessarily a connected component of $G-M'$),  whose vertices have $K$-monotone neighborhoods, and a particular vertex $r\in V(T)$, such that for every other vertex $u\in V(T)\setminus\{r\}$, $N(u) \subseteq K\cup V(T)$.
    If there are two vertex disjoint paths in $T -\{r\}$ with endpoints $a,b,c,d$ adjacent to $v_1$ (i.e. such that $1\le q_a,q_b,q_c,q_d$), delete $v_1$ and decrease the parameter $k$ by one.
\end{enumerate}
\begin{proof} (Safeness of \ruleref{e:r5pd})
Let us show that if $G$ has a feedback vertex set of size $k$ it has one containing~$v_1$. Consider a minimum feedback vertex set $X$ not containing $v_1$. The two cycles formed by the two paths and $v_1$ intersect only on $v_1$, thus $X$ must intersect each of these paths. Let $x,x'\in X$, be two vertices of these paths. We define $X'$ as obtained by replaying in $X$ the vertices $x$ and $x'$ with the at most two vertices of $K\setminus X$. Note that now $K\subseteq X'$, and in particular $v_1\in X'$. Since $(G- (V(T)\setminus \{r\})) - X'$ is a forest (as $X'\cup\{x,x'\}$ is a feedback vertex set), and since $T$ is a tree where only one vertex, $r$, has neighbors outside $V(T)\cup K$, the graph $G -X'$ is also a forest. Rule \ruleref{e:r5pd} is thus safe.\cqed \qed
\end{proof}

\begin{enumerate}[(R1)]
\setcounter{enumi}{5}
    \item \label{e:r6pd} Consider a $q$-tuple $K=(v_1,\ldots,v_{q})\in \ktypes{}$, a vertex $r$ of $G-M'$ adjacent to three leaves $a,b,c$ of $G-M'$, that have $K$-monotone neighborhoods. If their $K$-indices are such that $1\le q_a \le q_b \le q_c$, delete $a$ and maintain the parameter $k$. 
\end{enumerate}
\begin{proof}
Let us show that if $G - a$ has a feedback vertex set of size $k$, $X$, so does $G$. Inducing a clique, the set $K_a = (v_1,\ldots, v_{q_a})$ has at most two vertices not in  $X$. While possible, replace any vertex of $V(T)\setminus \{r\}$ in $X$, by a vertex of $K_a$. After this, we have every vertex of $K_a$ in $X$, or we have one vertex in $K_a\setminus X$, $v_i$, and none of the vertices of $V(T)\setminus \{r\}$ in $X$. In the latter case, we necessarily have $r\in X$. In both cases, the set $X$ is a feedback vertex set for $G$ also. Indeed, if $K_a\subseteq X$, vertex $a$ has at most one neighbor in $G - M'$, $r$, so it cannot create a cycle. If $K_a\subseteq X=\{v_i\}$ and $r\in X$, we also have that vertex $a$ has at most one neighbor in $G - M'$, $v_i$ here, so it cannot create a cycle. Rule \ruleref{e:r6pd} is thus safe.\cqed \qed
\end{proof}

Note that from now on, for every $K\in \ktypes{}$ and every vertex $v\in G-M'$, $v$ is adjacent to at most two leaves of $G-M'$ that have $K$-monotone neighborhoods.

\begin{enumerate}[(R1)]
    \setcounter{enumi}{6}   
    \item \label{e:r7pd} Consider a $q$-tuple $K=(v_1,\ldots,v_{q})\in \ktypes{}$, and a path $P$ in $G-M'$ of length six. If every vertex $u\in V(P)$ has degree two in $G-M'$, and has a $K$-monotone neighborhood in $M'$, delete $v_1$ and decrease the parameter $k$ by one.
\end{enumerate}
\begin{proof} (Safeness of \ruleref{e:r7pd}) 
Let us prove that if $G$ has a feedback vertex set of size~$k$, it has one containing~$v_1$. Let $X$ be a minimum feedback vertex set of $G$.
By \ruleref{e:r2pd}, $v_1$ is adjacent to all the six vertices in $P$. 
The graph induced by $\{v_1\}\cup V(P)$ having three cycles intersecting only at $v_1$, we have that $v_1\in X$ or that $P$ has at least three vertices in $X$. In the first case, we are done. In the second case, we replace any two vertices of $V(P)\cap X$ by the (at most two) vertices in $K\setminus X$ (thus including $v_1$) and let us note that $V(P)\cap X$ still contains at least one vertex.
Since $X$ is a feedback vertex set of $G - V(P)$ (as $G - V(P) - X$ is a subgraph of the forest obtained from the initial feedback vertex set), and since in $G - X$ the vertices of $V(P)\setminus X$ induce paths that each connect with at most one edge to a vertex outside $V(P)$, we have that $G - X$ is indeed a forest.\cqed \qed
\end{proof}

The last three rules deal with vertices having $(K,K')$-monotone neighborhoods with $K, K'\in \ktypes{}$.

\begin{enumerate}[(R1)]
\setcounter{enumi}{7}
    \item\label{e:r8pd} Consider a $q$- and a $q'$-tuple,  $K=(v_1,\ldots,v_{q})$ and $K'=(v'_1,\ldots,v'_{q'})$ of $\ktypes{}$, and a length 60 path $P$ in $G-M'$. If every vertex of $P$ has degree two in $G-M'$, and has a $(K,K')$-monotone neighborhood, delete $v_1$ and $v'_1$, and then decrease the parameter $k$ by two.
\end{enumerate}
\begin{proof} (Safeness of \ruleref{e:r8pd}) 
Let us prove that if $G$ has a feedback vertex set of size $k$, it has one containing $v_1$ and $v'_1$. Let $X$ be a minimum feedback vertex set of $G$. First, note that if $X$ has at least 5 vertices in $P$, then we can replace four of them in $X$ with the vertices of $K\setminus X$, and of $K'\setminus X$.  This would result in another minimum feedback vertex set.

By \ruleref{e:r7pd} $v_1$ has at least 10 neighbors in $P$ (otherwise \ruleref{e:r7pd} would apply with $v'_1$ and a subpath of $P$) so the graph induced by $V(P)\cup \{v_1\}$ contains five cycles intersecting only on $v_1$. Thus, a feedback vertex set avoiding $v_1$ would intersect $P$ on at least five vertices, a contradiction. Thus, we obtain that $v_1 \in X$. We similarly obtain that $v'_1 \in X$.\cqed \qed
\end{proof}

We already bounded the number of vertices in trees $T$ such that $\numtypes(T)\le 1$.
We now focus on trees $T$ such that $\numtypes(T) = 2$.
Given two tuples $K,K'\in \ktypes{}$, let $F_{K,K'}$ be the subforest of $G - M' - I_{M'}$ obtained by keeping the connected components $T$ such that $\types(T)=\{K,K'\}$.

A tree $T$ in $F_{K,K'}$ is said \emph{trivial} if it has only one vertex, say $u$. Clearly, this vertex has a $(K,K')$-monotone neighborhood.
By rule \ruleref{e:r0pd} presented in \autoref{sssec:inM'},  the number of trivial trees in $F_{K,K'}$ by $5p$ (as we have at most $p$ incomparable couples $(q_u,q'_u)$ of $K$- and $K'$-indices).
A tree $T$ in $F_{K,K'}$ is said \emph{semi-trivial} if it has only two vertices adjacent to $M'$, one having a $K$-monotone neighborhood, and the other having a $K'$-monotone neighborhood.
Note that, by rule \ruleref{e:r1pd} and \ruleref{e:r2pd}, $T$ has only two vertices.

\begin{enumerate}[(R1)]
\setcounter{enumi}{8}
    \item \label{e:r9pd} Consider a $q$- and a $q'$-tuple, $K=(v_1,\ldots,v_{q})$ and $K'=(v'_1,\ldots,v'_{q'})$ of $\ktypes{}$, and a semi-trivial connected component $ab$ of $F_{K,K'}$, such that vertex $a$ has $K$-monotone neighborhood, with $K$-index $q_a$, and such that vertex $b$ has $K'$-monotone neighborhood, with $K'$-index $q'_b$. If $F_{K,K'}$ has five other semi-trivial connected components, each of them with $K$-index at least $q_a$, and with $K'$-index at least $q'_b$, delete $a$ and $b$ and maintain the parameter $k$.
\end{enumerate}
\begin{proof}
    (Safeness of \ruleref{e:r9pd}) Let $G' = G-\{a,b\}$, and note that it cannot have a minimum feedback vertex set larger than $G$. Hence, it is sufficient to show that $G'$ has a minimum feedback vertex set $X$, that is also a feedback vertex set for $G$. 

    Let $A = N_G(a)\setminus \{b\} =\{v_1,\ldots ,v_{q_a}\}$ and $B= N_G(b)\setminus \{a\} =\{v'_1,\ldots ,v'_{q'_b}\}$, and let $a_i,b_i$ be the vertices of $m_i$ that have $K$- and $K'$-monotone neighborhood, respectively. Consider any minimum feedback vertex set $X$ of $G'$, and let us transform it (if needed) into a minimum feedback vertex set including all the vertices of $A\cup B$, except possibly one. Observe first that $G[K]-X$ and $G[K']-X$ being acyclic, we have that $|K \setminus X| \le 2$ and $|K' \setminus X| \le 2$.
    
    If $|A\setminus X| = 2$, then each of $\{a_1,\ldots,a_5\} \subseteq X$, as otherwise the two vertices in $A\setminus X$ and $a_i$ would form a triangle in $G' - X$. In such case, defining $X'$ by replacing in $X$ the vertices $a_1,\ldots,a_5$ with the vertices in $K\setminus X$ and in $K'\setminus X$ would result in the desired feedback vertex set. Indeed, $m_1,\ldots,m_5$ form isolated components in $G-X'$, and $G'-\{a_1,\ldots,a_5\}-X'$ is a forest, as $X \subseteq X' \cup \{a_1,\ldots,a_5\}$.
    
    Hence, we consider that $|A\setminus X|\le 1$ and $|B\setminus X|\le 1$. If $|A\setminus X| = 1$ and $|B\setminus X|= 1$, then $X$ contains at least four vertices among $\{a_1,\ldots,a_5\}\cup \{b_1,\ldots,b_5\}$, as otherwise there would be a 6-cycle going through $a_ib_i$ and $a_jb_j$ in $G' \setminus X$, for some $i,j$. In such case, defining $X'$ by replacing in $X$ the four vertices in $X\cap(\{a_1,\ldots,a_5\}\cup \{b_1,\ldots,b_5\})$ with the vertices in $K\setminus X$ and in $K'\setminus X$ would result in the desired feedback vertex set.
    Indeed, $m_1,\ldots,m_5$ form isolated components in $G'-X'$, and $G'-(\{a_1,\ldots,a_5\}\cup \{b_1,\ldots,b_5\}\cup X')$ is a forest, as $X \subseteq \{a_1,\ldots,a_5\}\cup \{b_1,\ldots,b_5\} \cup X'$.

    We thus have a minimum feedback vertex set of $G'$ including all the vertices of $A\cup B$, except possibly one. Since $ab$ has zero or one incident edge towards $G' - X$, we have that $G - X$ is a forest. \cqed \qed
\end{proof}
As for trivial trees, this rule allows us to bound the number of semi-trivial trees in $F_{K,K'}$ by $5p$ (as we have at most $p$ incomparable couples $(q_a,q'_b)$ of $K$- and $K'$-indices).

\begin{enumerate}[(R1)]
\setcounter{enumi}{9}
    \item \label{e:r10pd} If there are two tuples $K,K'\in \ktypes{}$, such that $F_{K,K'}$ 
     has four disjoint paths with endpoints adjacent to $v_1\in K$, delete $v_1$ and decrease the parameter $k$ by one.
\end{enumerate}
\begin{proof} (Safeness of \ruleref{e:r10pd}) 
Indeed, every feedback vertex set $X$ avoiding $v_1$ should intersect $F_{K,K'}$ on at least four vertices, but replacing these four vertices by the vertices of $K\setminus X$ and $K'\setminus X$, would result in another feedback vertex set 
with at most the same size, and containing $v_1$. We can thus force $v_1$ in the feedback vertex set.\cqed \qed
\end{proof}
Note that from now on, for every  $K,K'\in \ktypes{}$, the forest $F_{K,K'}$ has at most 6 connected components that are neither trivial, nor semi-trivial.

\begin{claim}
    \label{cl:wdeg2-forest}
    For every pair $K,K'\in \ktypes{}$, the forest $F_{K,K'}$ can be partitioned in two parts, one with at most $\O(p)$ vertices, and the other with at most six trees.
\end{claim}

In the following subsubsection, this claim will help us in bounding the size of the graph after the kernelization.

\subsection{Properties of the kernel}\label{sssec:notinM'size}

Let us now refer to the obtained graph and parameter as $G'$ and $k'$, and note by construction that $k'\le k$. 
As all rules are safe, we know that $(G',k')$ is equivalent to $(G,k)$.
Showing that $G'$ is a pseudo-disk graph with ply at most $p$ is simple. The class of pseudo-disk graphs with ply at most $p$ being closed under induced subgraphs, we only have to check that when applying Rule \ruleref{e:r2pd}, the graph remains a pseudo-disk graph with ply at most $p$, and this follows from \autoref{lem:tech-pseudo-disk-del-contr}. 
Now, it only remains to bound the size of $G'$.

Let us define the graph $G^*$ from a weak pseudo-disk system\footnote{That is, a system of regions homeomorphic to disks, and such that for any couple $\D_u,\D_v$, it is (only) required that the region $\D_u\setminus \D_v$ is connected.} obtained from the representation of $G'-M'$ by adding a small disk $\D_{v_P}$ for each element $P\in \alltypes$, and  for every pseudo-disk $\D_u$ of $B_{M'}$, by prolonging $\D_u$ along $\oHpp{u}$ in such a way to fulfill the following conditions (see \autoref{fig:prolonging-into-weak}).  \begin{enumerate}
    \item The prolongations do not intersect new pseudo-disks among $V(G')\setminus (M'\cup I_{M'})$. 
    \item The disks corresponding to $\alltypes$ are disjoint. 
    \item Any (prolonged) disk $\D_u$ of a vertex $u\in B_{M'}$ intersects any disk $\D_{v_P}$, such that $P\in \types(u)$ (possible by \autoref{cl:same-type-consecutive}). Furthermore, each component of $\D_{v_P}\cap \D_u$ is a lens, that is, it is bordered by the union of a sub-arc of $\partial\D_{v_P}$ and a sub-arc of $\partial\D_u$.
\end{enumerate}
Hence, we deduce that this intersection system is indeed a weak pseudo-disk one.
Furthermore, since this representation has ply two, the graph $G^*$ is planar \cite{kratochvil1996intersection}.

\begin{figure}
    \centering
    \includegraphics[width=\textwidth]{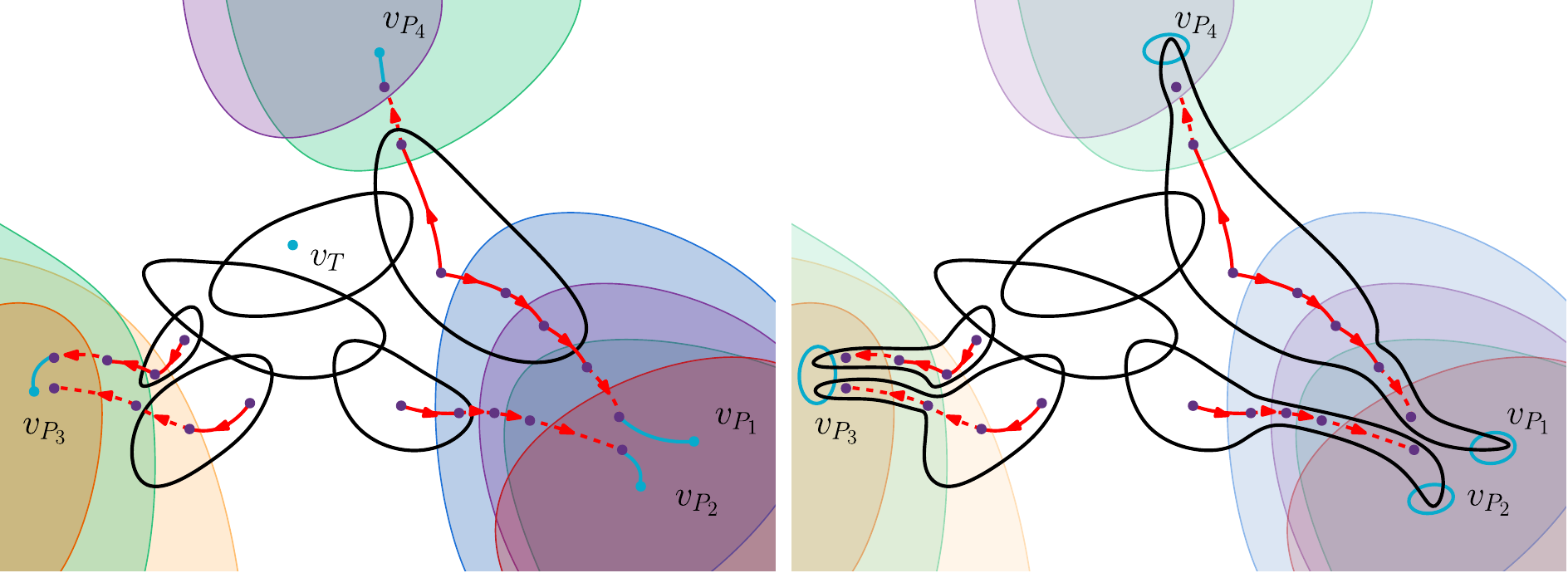}
    \caption{Construction of $G^*$.}
    \label{fig:prolonging-into-weak}
\end{figure}

As $|V(G')| \le |M'| + |V(G^*)|$, we are going to bound  $|V(G')|$ by bounding $|V(G^*)|$. The main  difficulty is to establish 
that for any tree $T$ of $G'-M'-I_{M'}$ we have
$|V(T)| = \O(\numtypes(T))$.
Before proving this in \autoref{cl:sizeTwdeg}, we have to introduce preliminary notions.

Consider an arbitrary connected component $T$ of $G - M'-I_{M'}$ such that $\types(T) = \{P_1,\ldots,P_d\}\subseteq \alltypes$, for some $d\ge 2$. Then, root $T$ at an arbitrary vertex~$r$. For any vertex $v\in V(T)$ denote by $T_v$ the rooted subtree of $T$ rooted at $v$
and let us denote by $\types^T(v)$ the set $\bigcup_{u\in V(T_v)} \types(u)$. Note that $\types^T(v) \subseteq \{P_1,\ldots,P_d\}$. The vertex of $G^*$ corresponding to a path $P\in \alltypes$ will be denoted $v_P$. Two vertices $u,v$ of $T$ are said \emph{incomparable} if $u\notin T_v$ and $v\notin T_u$.

The graph $G^*$ being planar we have the following two properties.
\begin{claim}\label{cl:k33}
    For any three, pairwise incomparable, vertices $u,v,w$ of $T$, there is no pair $i,j\in \{1,\ldots,d\}$ such that $\{P_i,P_j\}$ is a subset of $\types^T(u)$, $\types^T(v)$, and $\types^T(w)$.
\end{claim}
\begin{proof}
    Indeed, otherwise by contracting $T_u$, $T_v$, $T_w$, and the rest of $T$ on $u, v, w$, and $r$, respectively, these four vertices with the vertices $v_{P_i}$ and $v_{P_j}$ would form a $K_{3,3}$ minor in $G^*$.\cqed \qed
\end{proof}
\begin{claim}\label{cl:toomanyedges}
There are at most $3d$ vertices $v$ in $T$ that are pairwise incomparable and that have distinct pairs $\{P_i,P_j\}$ in their set $\types^T(v)$.
\end{claim}
\begin{proof}
    Indeed, otherwise one could construct a minor of $G^*$, with vertex set $\{v_{P_1},\ldots,v_{P_d}\}$ and more than $3d$ edges, which is impossible for a planar graph.\cqed \qed
\end{proof}

Let $L\subseteq V(T)$ ($L$ for large) be the set of vertices $v$ such that $|\Tilde{N}^T(v)| \ge 2$.
\begin{claim}\label{cl:TI-bounded}
    If $d\ge 2$, we have $|L| = \O(d)$.    
\end{claim}
\begin{proof}
    Observe that if $v \in L$ and $u\in V(T)$ is an ancestor of $v$ (i.e., $v \in T_u$), then $u \in L$.
    Also, as $d\geq 2$, $r\in L$. So $T[L]$ is connected and the leaves of $T[L]$ form a set of incomparable vertices.
    By \autoref{cl:k33} and \autoref{cl:toomanyedges}, we deduce that $T[L]$ has at most $6d$ leaves. This implies that the number of vertices of $T[L]$ with degree at least three (in $T[L]$) is also bounded by $6d$. The same bound of $6d$ holds on the number of connected components in the subgraph of $T[L]$ induced by the vertices having degree two in $T[L]$. In the following, let us refer to these connected components as \emph{long paths}.
    
    To bound the number of vertices of degree two (in $T[L]$), let us now cut each long path into order-60 subpaths whenever possible, and let us refer to these paths as the \emph{short paths}. Each long path leading to at most one short path of order smaller than 60, these short paths contribute in total for at most $6d\cdot 59 = \O(d)$ vertices in $T[L]$. Let us now bound the number of short paths of order exactly 60. By  \ruleref{e:r8pd}, for every such short path $P$ we have that $\numtypes(P)\ge 3$.
    This implies that the minor of $G^*$ obtained by contracting these short paths and keeping the vertices in $\{v_{P_1},\ldots,v_{P_d}\}$ has at most $2d$ contracted vertices (by \autoref{lem:planar-linear-size}).
    So, there are at most $2d$ short paths of length 60, and in total, the size of $L$ is bounded by $6d + 6d + (6d\cdot 59) + (2d\cdot 60) = \O(d)$.\cqed \qed
\end{proof}

\begin{claim}\label{cl:TminusI}
    For every vertex $v \in L$ and every $P \in \types^T(v)$, the subtree $T'$ of $T$ rooted at $v$ and containing every tree $T_u$ such that $u$ is a son of $v$ and $\types^T(u)=P$ has constant size.
\end{claim}
    \begin{proof}
    Let $K_\alpha = (u_1,\ldots,u_{q})$, for some $q\le p$, be the clique corresponding to the border $\alpha$. As all the leaves of $T'$ have a non-empty $K_\alpha$-monotone neighborhood. Rules~\ruleref{e:r5pd} and \ruleref{e:r6pd} imply that any vertex $v$ of $T'$ has at most three sons in $T'$. Indeed, with four sons we would have either three leaves among them, or two subtrees with at least two vertices each. In the first case, Rule~\ruleref{e:r6pd} should have been applied. In the second case, by Rule~\ruleref{e:r2pd}, the two non-trivial subtrees have at least two leaves each. Since the leaves are connected to $u_1$ (by Rule~\ruleref{e:r1pd}), these subtrees contain two disjoint paths connecting neighbors of $u_1$. In that case, Rule~\ruleref{e:r5pd} should have been applied. 
    
    Also, $T'$ has depth at most four, as otherwise $T'$ would contain two disjoint paths avoiding $v$ and such that their endpoints are adjacent to $u_1$, and in such case Rule~\ruleref{e:r5pd} should have applied. If there was a path $P=(v,t_1,\ldots , t_5)$, there would be four vertex disjoint paths from $t_1,t_2,t_3,t_4$ to leaves in $T' -\{v\}$. Indeed, for any $i\in \{1,2,3,4\}$, the vertex $t_i$ has degree at least three (by Rule~\ruleref{e:r2pd}), this third neighbor (distinct from $t_{i-1}$ and $t_{i+1}$, with $t_0=v$) is either $u_1$ or it belongs to $T'$, and we denote it $x_i$. The subtree of $T'$ rooted at $x_i$ contains a path linking $x_i$ to a leaf, which is adjacent to $u_1$. Hence, in any case we have the path from $t_i$ to a leaf. 
    
    Finally,  $T'$ having bounded degree and bounded depth, it has bounded size.\cqed \qed
\end{proof}

Let us define the graph $G^*_L$ from $G^*[V(T)\cup\{v_{P_1},\ldots,v_{P_d}\}]$ by contracting every edge $uv$ such that $|\types^T(u)|=1$. 
As a minor of $G^*$, this graph is planar and has $\O(d)$ vertices by \autoref{cl:TI-bounded}
(since $V(G^*_L) = L\cup\{v_{P_1},\ldots,v_{P_d}\}$). Being planar, its number of edges is also $\O(d)$. 
Given a vertex $v \in L$ and a path $P \in \types^T(v)$, such that $v$ has a son $u$ with $\types^T(u)=P$, then
$v$ and $v_P$ are adjacent in $G^*_L$. Hence, the number of vertices of $V(T)\setminus L$ is bounded by the number of edges in $G^*_L$, times the constant of \autoref{cl:TminusI}. This set has thus also size $\O(d)$.
We can thus conclude the following.
\begin{claim}
    \label{cl:sizeTwdeg}
    For any tree $T$ of $G'-M'-I_{M'}$ such that $\numtypes(T)\ge 2$, we have $|V(T)| = \O(\numtypes(T))$.
\end{claim}

By \autoref{cl:Dv-inside-bounded}, $|I_{M'}|=\O(p^4k)$.
By \autoref{cl:few-wedg}, the set $\bigcup_{T} V(T)$, where the union is taken over all the trees $T\in G'-M'-I_{M'}$ such that $\numtypes(T)\le 1$, has size $\O(p|M'|)=\O(p^3k)$.
By \autoref{cl:few-degree-two-T}, \autoref{cl:wdeg2-forest}, and \autoref{cl:sizeTwdeg}, the set $\cup_{T} V(T)$, where the union is taken over all the trees $T\in G'-M'-I_{M'}$ such that $\numtypes(T) =2$, has size $\O(p^2|M'|) = \O(p^4k)$.
By \autoref{cl:few-large-degree-T} and \autoref{cl:sizeTwdeg}, the set $\bigcup_{T} V(T)$, where the union is taken over all the trees $T\in G'-M'-I_{M'}$ such that $\numtypes(T)\ge 3$, has size $\O(p|M'|)=\O(p^3k)$. Hence, the total number of vertices in $G'$ is $\O(p^4k)$.
This concludes the proof of \autoref{lemma:red-fvs-pseudo}.

\section{Contact-segment graphs as pseudo-disk graphs}
\label{sec:contactsegp-arepseudo}
Intersection graphs of segments in $\RR^2$ are called \emph{segment graphs}. 
In a segment contact representation, any point belonging to two segments must be an endpoint of at least one of these segments.  
If a point belongs to several segments, the above property must hold for every pair of them. 
This defines \emph{\CONTACTSEG graphs}.
This class of graphs is a subclass of pseudo-disk graphs. This, up to the authors knowledge, does not appear in the literature. We hence prove it formally here.
\begin{theorem}\label{thm:conact-seg-are-pd}
    Every \CONTACTSEG graph is a pseudo-disk graph.
\end{theorem}
\begin{proof}
Consider any contact system of segments $\mS$, and let $\epsilon >0$ be a sufficiently small value such that (1) every segment has length at least $\epsilon$, and (2) any two points in two non-intersecting segments are at distance more than $2\epsilon$ from each other.

The construction of the equivalent system of pseudo-disks goes in two steps. First, let us shorten every segment by $\epsilon/2$. We denote $\mS'$ the obtained system of segments. Note that in $\mS'$ all the segments are disjoint, but that $a',b'\in \mS'$ are at distance at most $\epsilon$, if and only if the corresponding original segments in $\mS$, $a,b$, were intersecting. Secondly, for each segment $a'\in \mS'$, let $\D_a$ be the set of points at distance at most $\epsilon$ from a point of $a'$. Clearly, every set $\D_a$ is homeomorphic to a disk, and any two such sets $\D_a,\D_b$ intersect if and only if the segments $a,b\in \mS$ intersect. Hence, we thus just have to show that if $\D_a$ and $\D_b$ intersect, then  $\pD_a$ and $\pD_b$ intersect on two points exactly. Let us consider the cases where $a$ and $b$ intersect at their endpoints, and the case where their intersection is an inner point for one of them. In both cases, \autoref{fig:contact-seg-pseudo-disks} clearly shows that there are only two intersection points.    \qed
\end{proof}
\begin{figure}
    \centering
    \includegraphics[width=0.65\textwidth]{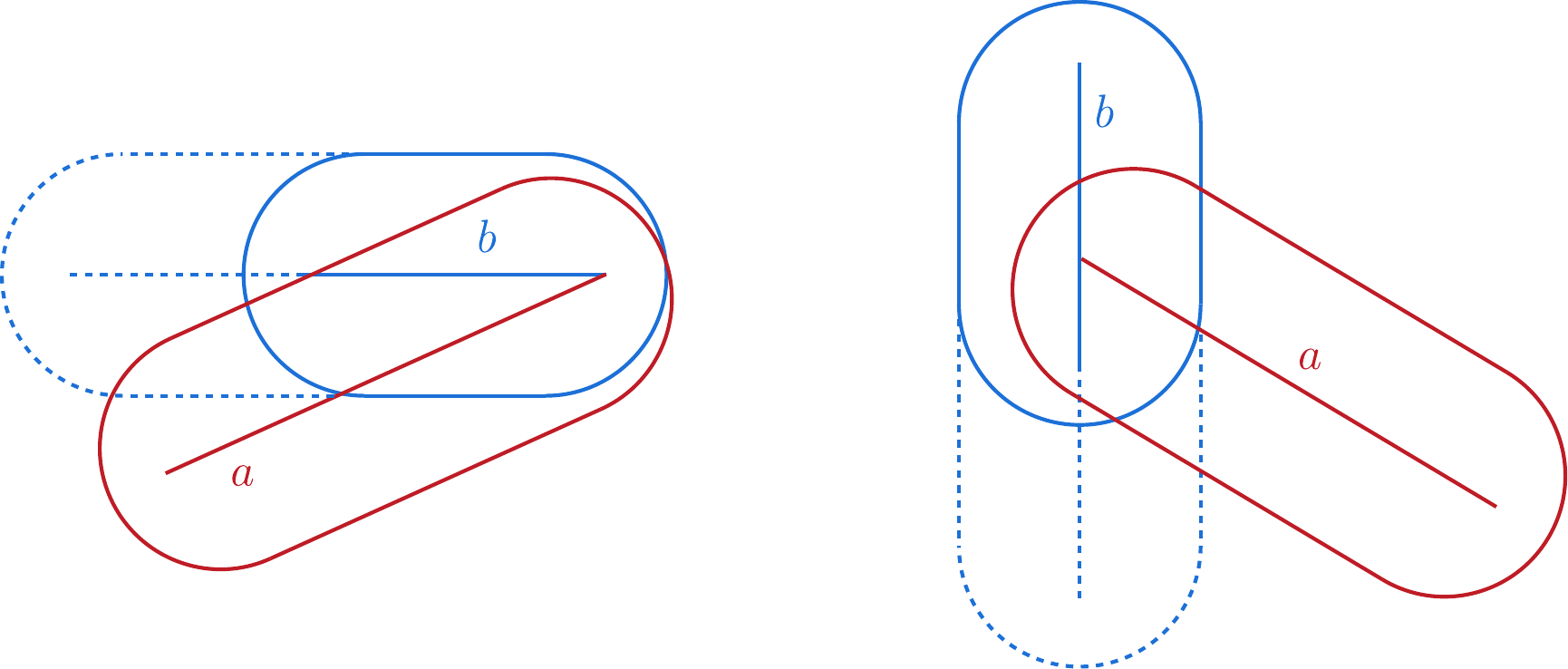}
    \caption{(left) Two segments $a,b$ intersecting at their endpoints, and the corresponding pseudo-disks. According to the lengths of $a$ and $b$, and the angle they form, the intersection points of $\pD_a$ and $\pD_b$ occur on their straight side or on their rounded extremity. (right) Two segments $a,b$ intersecting at an endpoint of $a$, and inside $b$. The corresponding pseudo-disks with border $\pD_a$ and $\pD_b$ intersecting on two points.}
    \label{fig:contact-seg-pseudo-disks}
\end{figure}

\begin{theorem}\label{thm:FVS-contact-seg}
    There is an algorithm that, given a $n$-vertex \CONTACTSEG graph with a representation and a parameter $k$, solves \FVS in time $2^{\O(k^{5/6} \log k)}n^{O(1)}$.
\end{theorem}
\begin{proof}
The algorithm is the same as the one of \autoref{thm:FVS-pseudo}, but the analysis changes in a few places. First, one should note that given an instance $(G,k)$ the branching at the beginning (obtained from \autoref{cor:bothbranchings}) produces a set $\C$ of instances $(G',M,k')$, where $G'$ is an induced subgraph of $G$, we can hence assume that $G'$ is provided with a \CONTACTSEG representation. Let us see now where the proof of \autoref{thm:FVS-pseudo} can be improved.

Note that by taking a sufficiently small $\epsilon$ in the proof of \autoref{thm:conact-seg-are-pd} one gets a pseudo-disk representation where every pseudo-disk contains a point of ply one. This implies that the set of inner-$M'$-vertices is necessarily empty, that is $|I_{M'}|=0$.

Note that in the \CONTACTSEG representation of $G'$, there are at most $2|M|$ points where segments of $M$ intersect, and all the segments going through such point but at most one belong to $M$ (because $M$ hits every triangle at least twice). Furthermore, every segment of $M'\setminus M$ intersects at least one of these points (assuming again that $\epsilon$ is sufficiently small in the construction of \autoref{thm:conact-seg-are-pd}). Hence, we have that $|M'|=\O(|M|)=\O(pk)$.

In the proof of \autoref{lemma:red-fvs-pseudo} it is shown that after the kernelization, the reduced instance has size $|I_{M'}|+p^2 |M'|$ (last paragraph of the proof), which is $\O(p^3k)$ for \CONTACTSEG graphs. Plugging this value in the proof of \autoref{thm:FVS-pseudo}, one gets the claimed complexity. \qed
\end{proof}

A finer analysis could probably lead to a time complexity of 
 $2^{\O(k^{4/5} \log k)}n^{O(1)}$. In any case, it still misses the lower bound of $2^{\o(\sqrt{n})}$ obtained for \CONTACTDEUXDIR\xspace  graphs, that are \CONTACTSEG graphs where the segments are either horizontal or vertical.

\section{Robust algorithms for \TH and \FVS in pseudo-disk graphs.}
\label{sec:robust}

In this section, we explain how to generalize the robust algorithms for \TH and \FVS in disk graphs of \cite{Faster2023Shinwoo} for pseudo-disk graphs. The first obstacle for generalization is that both algorithms use an EPTAS~\cite{Bonamy21EPTASCLIQUES} 
for the clique problem in disk graphs, and such algorithm is not known for pseudo-disk graphs. They use this EPTAS for the following branching process. If there is a clique of size $p$, for every subset of $p-2$ vertices, create a new instance where these vertices have been forced in the hitting set.
Actually, to apply this branching (without altering its time complexity), one only needs an algorithm running in time $2^{\O(\frac{k}{p}\log p)}n^{\O(1)}$ finding a clique of size $p$, if such clique exist. 
Achieving such time complexity with $k=1$ is impossible in general, but recall that the parameter $k$ comes from our instance of \TH or \FVS. We thus restrict ourselves to graphs $G$ for which the minimum triangle hitting set (resp. minimum feedback vertex set) is not obviously above $k$ (otherwise, we do not need to continue the branching process, and we just output that this is a no-instance).

\begin{lemma}\label{lem:p-clique-th-fvs}
    Given a graph $G$ and integers $3\leq p \leq k$, either $G$ does not have a feedback vertex (respectively a triangle hitting set) of size at most $k$ or we can find in time $2^{\O(p\log k)}n^{\O(1)}$ a clique of size $p$ in $G$ if it exists.
\end{lemma}
\begin{proof}
    We start by computing a $2$-approximation $S$ for the \FVS problem in $G$ in polynomial time using the algorithm in \cite{2approx-fvs-1} or \cite{2approx-fvs-2} (respectively a $3$ approximation $S$ for the \TH problem using a greedy algorithm). If $|S|>2k$ (respectively $|S|>3k$) we know that $G$ does not have a feedback vertex set (respectively a triangle hitting set) of size at most $k$. Otherwise suppose $G$ contains a clique $C$ of size $p$ in $G$. We know that $|C\setminus S|\leq 2$ as otherwise there would be a triangle not hit by the feedback vertex set $S$ (respectively the triangle hitting set) which is not possible. Guessing $p-2$ vertices of $C$ in $S$ and then the two remaining vertices in $V(G)$ takes time $|S|^{p-2}n^2=2^{\O(p\log k)}n^{\O(1)}$.\qed
\end{proof}
Notice that for the branching process we only look for cliques of size $p$, with $p\le k+3$, as a clique of size $k+3$ is already a certificate of a no-instance.
Furthermore, if $p$ is chosen such that $p=\O(\sqrt{k})$, the algorithm of \autoref{lem:p-clique-th-fvs} runs in time $2^{\O(\frac{k}{p}\log p)}n^{\O(1)}$, and with this value the branching holds, with the same time-complexity.
After this branching process, the algorithm for \TH and for \FVS differ.

For the case of triangle hitting a core (see \cite{Faster2023Shinwoo} for the definition) of size $\O(pk)$ can be found in the same manner as in the original article (using \autoref{lem:pseudo-number-edges} to obtain the bound on the size of the core), giving a generalization to pseudo-disks with the same complexity as the case of disk graphs.

For the \FVS problem, the running time of the robust algorithm is bounded using the following key lemma from \cite{lokSODA22}. It was stated for disk graphs, but the proof only uses the fact that disks are pseudo-disks and so the generalization holds. 
\begin{lemma}{(Stated for disk-graphs in \cite{lokSODA22})}
    Let $G$ be a pseudo-disk graph that has some realization of ply $p$, and let $M\subseteq V(G)$ be such
that for all $v \in M,~N_G(v) \setminus M$ is an independent set, there does not exist a vertex in $V (G) \setminus M$ whose neighborhood is contained in $M$, and $G-M$ has treewidth $w$. Then, the treewidth of $G$ is $\O(max\{\sqrt{|M| \cdot w} \cdot p^{2.5}, w \cdot p\})$.
\end{lemma}

But being able to bound the size of the cliques and this key lemma are not enough, a last obstacle to generalize the robust algorithm of \cite{Faster2023Shinwoo} is to bound the number of ``deep and regular'' vertices (see \cite{Faster2023Shinwoo} for the definition). They are exactly the inner-$M'$-vertices whose hosted graphs are paths in our proof, and can be bound in the same way as in \autoref{cl:Dv-inside-bounded} using \autoref{hyperpseudo}. However, in the  Lemma 8 of \cite{Faster2023Shinwoo}, the number of such vertices in disk graphs is bounded by $\O(p^2k)$ while \autoref{hyperpseudo} gives the bound $\O(p^3k)$ for pseudo-disk graphs, resulting in a worse complexity.

The obtained results are summarized in \autoref{fig:results}.
\begin{table}[!ht]
\begin{center}
\renewcommand{\arraystretch}{1.4}
\begin{tabular}{|c|c|c|c|c|}
    \hline
     Class & Problem & Time complexity & Robust & \\
     \hline
     \multirow{4}*{Disk} & \multirow{2}*{\FVS} & $2^{\O(k^{7/8}\log k)}n^{\O(1)} $& No & \multirow{4}*{\cite{Faster2023Shinwoo}}\\ \cline{3-4}
      &  & $2^{\O(k^{9/10}\log k)}n^{\O(1)} $& Yes & \\ \cline{2-4}
      & \multirow{2}*{\TH} & $2^{\O(k^{2/3}\log k)}n^{\O(1)} $& No & \\ \cline{3-4}
      &  & $2^{\O(k^{3/4}\log k)}n^{\O(1)}$\protect\footnotemark& Yes & \\ \cline{1-5}
     \multirow{3}*{Pseudo-disk} & \multirow{2}*{\FVS} & $2^{\O(k^{6/7}\log k)}n^{\O(1)} $& No & \autoref{sec:fvspseu}\\
     \cline{3-5}
      &  & $2^{\O(k^{10/11}\log k)}n^{\O(1)} $& \multirow{2}*{Yes} & \multirow{2}*{\autoref{sec:robust}}\\\cline{2-3}
      & \TH & $2^{\O(k^{3/4}\log k)}n^{\O(1)} $&  & \\
     \cline{1-5}
     \multirow{2}*{Contact-segment} & \multirow{2}*{\FVS} & $2^{\O(k^{5/6}\log k)}n^{\O(1)} $& No & \autoref{sec:contactsegp-arepseudo}\\
     \cline{3-5}
      &  & $2^{\O(k^{7/8}\log k)}n^{\O(1)} $& \multirow{1}*{Yes} & \cite{berthe2023subexponential}
     \\\cline{1-5}

     \hline
\end{tabular}
\renewcommand{\arraystretch}{1}
\end{center}
\caption{Summary of the previous and new results.}
\label{fig:results}
\end{table}
\footnotetext{The published version of the paper gives a bound of $2^{O(k^{4/5} \log k)}n^{O(1)}$ but it can easily be improved to $2^{O(k^{3/4} \log k)}n^{O(1)}$, as confirmed to us by the authors of \cite{Faster2023Shinwoo} (private communication).}

Here we have only considered the generalizations of the robust algorithms of \cite{Faster2023Shinwoo}, as their analysis for the complexities of the algorithms taking a representation are mainly based on the result that disk graphs have balanced separators of clique-size $\O(\sqrt{n})$, but for pseudo-disk graphs only separators of clique-size $\O(n^{2/3}\log n)$ are known (see \cite{Clique2023Berg} for the definition of the clique-size of a separator, and the mentioned bounds), giving worse complexities than the robust cases.

\section{Discussion}\label{sec:disc}
In this paper, we provide an algorithm for solving \fvs in pseudo-disk graphs running in times $2^{\O(k^{6/7}\log k)}n^{\O(1)}$. This generalizes the previous results~\cite{Faster2023Shinwoo,Lokshtanov23Approx} on disk graphs, with an improvement on the running time.
On the other hand, as noted in the introduction, there is for both problems an ETH-based lower-bound of $2^{o(\sqrt{n})}$~\cite{berthe2023subexponential}. So a natural problem is to get matching upper- and lower-bounds, in particular for \CONTACTSEG graphs. We have no evidence to believe that our upper-bounds could be tight.
Besides, our algorithm requires a pseudo-disk representation of the input graph. So a second open problem is to provide a robust subexponential FPT algorithm for \FVS in pseudo-disk graphs.
The bulk of our algorithm for \FVS consists in applying reduction rules to obtain an instance of size polynomial in the parameter. However, this is not strictly speaking a kernel as we do not reduce the input instance but the subexponential number of instances produced by the preprocessing step. It could however be interesting to investigate if the ideas from our reduction steps could be useful for kernelization in pseudo-disk graphs.

\bibliography{biblio}

\end{document}